\documentclass[aps,pra,superscriptaddress,notitlepage,twocolumn]{revtex4-1}

\usepackage[utf8]{inputenc}
\usepackage{amsmath}
\usepackage{stmaryrd}
\usepackage{color}
\usepackage{mathrsfs}
\usepackage{yfonts}
\usepackage{rotating}
\usepackage{graphicx}
\usepackage{amsfonts}
\usepackage{mathrsfs}
\usepackage{amsmath}
\usepackage[american]{babel}
\usepackage{wasysym}
\usepackage{slashed}
\usepackage{pgfplots}
\usepackage{braket}
\usepackage{graphicx}
\usepackage[caption=false]{subfig}
\usepackage{dsfont}
\usepackage{amsthm}
\usepackage{youngtab}

\usepackage{natbib}
\PassOptionsToPackage{breaklinks}{hyperref}
\usepackage[colorlinks,allcolors=black]{hyperref}

\theoremstyle{definition}

\newtheorem{remark}{Remark}
\newtheorem{lemma}{Lemma}
\newtheorem{observation}{Observation}
\newtheorem{proposition}{Proposition}
\newtheorem{theorem}{Theorem}
\newtheorem{corollary}{Corollary}

\newtheorem{procedure}{Procedure}

\usepackage{cleveref}
\crefname{proposition}{Proposition}{Propositions}
\crefname{theorem}{Theorem}{Theorems}
\crefname{definition}{Definition}{Definitions}
\crefname{lemma}{Lemma}{Lemmas}
\crefname{figure}{Figure}{Figures}
\crefname{corollary}{Corollary}{Corollary}
\crefname{conjecture}{Conjecture}{Conjectures}
\crefname{section}{Section}{Sections}
\crefname{appendix}{Appendix}{Appendixes}
\crefname{observation}{Observation}{Observation}
\crefname{remark}{Remark}{Remark}
\crefname{example}{Example}{Examples}
\crefname{equation}{Eq.}{Eqs.}
\crefname{table}{Table}{Tables}
\crefname{procedure}{Procedure}{Procedures}

\newcommand{\blue}[1]{\textcolor{blue}{#1}}

\usepackage{tikz}

\usepackage{amssymb}
   \makeatletter
     \renewcommand\@make@capt@title[2]{%
      \@ifx@empty\float@link{\@firstofone}{\expandafter\href\expandafter{\float@link}}%
       {\textbf{#1}}\@caption@fignum@sep#2\quad}%
     \makeatother
\makeatletter
\renewcommand{\fnum@figure}{\textbf{Figure~\thefigure}}

\usepackage{mathtools}
\newcommand{\SL}[0]{\text{SL}}
\newcommand{\SLC}[0]{\text{SL(}2,\mathbb{C}\text{)}}
\newcommand{\SLIP}[0]{\text{SLIP}}

\newcommand{\be}[0]{\beta}

\newcommand{\si}[0]{\sigma}

\newcommand{\Id}[0]{\text{Id}}

\newcommand{\bs}[1]{\textbf{#1}}
\newcommand{\ea}[1]{\begin{align}#1\end{align}}
\newcommand{\eq}[1]{\begin{equation}#1\end{equation}}

\newcommand{\ma}[1]{\mathcal{#1}}

\usepackage{cancel}

\begin{document}
\title{
Entanglement Classification via Single Entanglement Measure
}
\author{Adam Burchardt}
\email{adam.burchardt.uam@gmail.com}
\affiliation{QuSoft, CWI and University of Amsterdam, Science Park 123, 1098 XG Amsterdam, the Netherlands}
\author{Gon\c{c}alo M. Quinta}
\affiliation{Instituto de Telecomunica\c{c}\~{o}es, Physics of Informations and Quantum Technologies Group, Lisboa, Portugal}
\author{Rui Andr\'{e}}
\affiliation{Centro de Astrof\'{\i}sica e Gravita\c c\~ao  - CENTRA,
Departamento de F\'{\i}sica, Instituto Superior T\'ecnico - IST,
Universidade de Lisboa - UL, Av. Rovisco Pais 1, 1049-001 Lisboa, Portugal}


\begin{abstract}

We show that a single polynomial entanglement measure is enough to verify equivalence between generic $n$-qubit states under Stochastic Local Operations with Classical Communication (SLOCC). SLOCC operations may be represented geometrically by Möbius transformations on the roots of the entanglement measure on the Bloch sphere. Moreover, we show how the roots of the 3-tangle measure classify 4-qubit generic states, and propose a method to obtain the normal form of a 4-qubit state which bypasses the possibly infinite iterative procedure.
\end{abstract}

\maketitle


\section{Introduction}
\label{sec:Introduction}
Quantum entanglement is one of the key manifestations of quantum mechanics and the main resource for technologies founded on quantum information science. In particular, quantum states with non-equivalent entanglement represent distinct resources which may be useful for different protocols. The idea of clustering states into classes exhibiting different qualities under quantum information processing tasks resulted in their classification under stochastic local operations assisted by classical communication (SLOCC). Such a classification was successfully presented for two, three and four qubits \cite{ThreeQub, FourQubits, SLOCCallDim,4QubitsAllClasses}.
However, the full classification of larger systems is completely unkown. Even the much simpler problem of detecting if two $n$-qubit states ($n>4$) are SLOCC-equivalent is, in general, quite demanding \cite{PolInv4qubits,Zhang_2016,KempfNessToEntanglement,BurchardtRaissi20}.

Among several approaches to the entanglement quantification and classification problem, a particularly useful one is via \textit{$\SL$-invariant polynomial} ($\SLIP$) measures.
Well-known examples are concurrence and $3$-tangle, which measure the 2-body and 3-body quantum correlations of the system \cite{PhysRevLett.80.2245,DistributedEntanglement}.
$\SLIP$ measures provide not only a convenient method for entanglement classification but also its practical detection. Indeed, it was shown that almost all SLOCC equivalence classes can be distinguished by ratios of such measures \cite{SLOCCallDim}. Any given two n-qubit states are then SLOCC-equivalent if a \textit{complete set of SLIP measures} has the same values for both of them \cite{PolInv4qubits}. For more than four qubits, however, the size of such a set grows exponentially, making it intractable to use this approach to discriminate SLOCC-equivalent states with more than four qubits \cite{Love07}.

In this paper, we show that, contrary to intuition, a single $\SLIP$ measure is enough to verify the SLOCC-equivalence between any two generic pure $n$-qubit states. By generic state, we mean the set of all pure states except of the measure zero subset with respect to the Haar measure. We achive this by using a mathematical trick, in particular, we use a $\SLIP$ measure defined for $n$ qubit systems to verify equivalence of systems of $n+1$ quibits. In essence, we look how the roots of the SLIP measure for those states behave under SLOCC. We show that if the states are SLOCC-equivalent, then the roots of the SLIP measure for each state must be related by a Möbius transformation, which is straightforward to verify. In particular, we use this procedure to show that the 3-tangle measure is enough to discriminate between generic 4-qubit states. Finally, we show how one may use the roots of a SLIP entanglement measure to obtain the normal form of a 4-qubit state which bypasses the possibly infinite iterative standard procedure. 

\section{Polynomial Invariant Measures}
\label{sec:Polynomial Invariant Measures}

An \textit{entanglement measure} is a function $E(\ket{\psi})$ defined for pure states of $n$ qubits which vanishes on the set of separable states.
One of the desired features of entanglement measures is invariance under SLOCC operations.
Mathematically, a SLOCC operation might be uniquely determined by the action of local invertible operators $L \in \SLC^{\otimes n}$ \cite{ThreeQub}.
An entanglement measure $E$ defined for all pure states of $n$ qubits is called a \textit{$\SL$-invariant polynomial of homogeneous degree $h$} if it is homogeneous polynomial of degree $h$ in the state coefficients, an it is invariant under any local operation $\mathcal{O}=\mathcal{O}_1\otimes\cdots\otimes \mathcal{O}_n$, where $\mathcal{O}_i \in \SLC$. 
It is easy to see that those two conditions are equivalent to the fact that $E$ satisfies
\begin{equation*}
E\big( \kappa\, \mathcal{O} \ket{\psi} \big) =\kappa^h  E\big(\ket{\psi} \big)
\end{equation*}
for each real constant $\kappa > 0$ and invertible linear operator $\mathcal{O}$~\cite{ThreeQub,SLOCCallDim,Eltschka_2014}. Such polynomial will be denoted by  $\SLIP_n^h$, where the upper index indicates the degree of the polynomial and the lower index is related to the number of qubits. Well-known examples are concurrence and $3$-tangle \cite{PhysRevLett.80.2245,DistributedEntanglement}. 
Both of those measures, concurrence and $3$-tangle, can be written as the absolute value of the anti-linear expectation value of simple operators. Indeed, for a two-qubit pure state $\ket{\psi}=c_{00}\ket{00} +c_{01}\ket{01}+c_{10}\ket{10}+c_{11}\ket{11}$ its concurrence reads as 
$
C(\ket{\psi})
=| c_{00}c_{11}-c_{01}c_{10}|
=
\big\vert \bra{\psi}\sigma_y \otimes \sigma_y \ket{\bar{\psi}} \big\vert
$. 
Furthermore, the 3-tangle $\tau^{(3)} $ defined for 3-qubit states $\ket{\psi}\in \mathcal{H}_2^3$ takes relatively simple form:
\begin{equation}
\label{3tangle}
\tau^{(3)} \big(\ket{\psi}\big)
=
\Bigg\vert 
\sum_{j=\Id,x,y,z} \eta_j 
\Big(
\bra{\psi}\sigma_j \otimes\sigma_y \otimes \sigma_y 
\ket{\bar{\psi}} 
\Big)^2
\Bigg\vert ,
\end{equation}
with notation $(\eta_\Id,\eta_x,\eta_y,\eta_z ) =( -1,0,1,1)$. Moreover, the degree-4 polynomial invariants for 4 qubits described by Luque and Thibon in Ref~\cite{PhysRevA.67.042303} can be also written as similar expressions \cite{PhysRevA.85.022301}. 
This simple idea of exploring the anti-linear expectation value of the tensor product of Pauli operators was further used for constructing invariants of an arbitrary number of qubits \cite{PhysRevA.72.012337,Eltschka_2014,PhysRevA.85.022301}. For example, the following formulas
\begin{align}
\label{(2k+1)}
E^{(2k+1)}& \big(\ket{\psi}\big)=
\\ 
\nonumber
=&
\Bigg\vert 
\sum_{j=\Id,x,y,z} \eta_j \;
\Big(
\bra{\psi}\sigma_j \otimes \underbrace{\sigma_y \otimes\cdots\otimes \sigma_y}_{2k} 
\ket{\bar{\psi}} 
\Big)^2
\Bigg\vert ,
\end{align}
\begin{align}
\label{(2k)}
E&^{(2k)} \big(\ket{\psi}\big)=
\\ 
\nonumber
=&
\Bigg\vert 
\sum_{j, i=\Id,x,y,z} \eta_j \; \nu_i \;
\Big(
\bra{\psi}\sigma_j \otimes \sigma_i \otimes\underbrace{\sigma_y \otimes\cdots\otimes \sigma_y}_{2k} 
\ket{\bar{\psi}} 
\Big)^2
\Bigg\vert ,
\end{align}
with $(\eta_\Id,\eta_x,\eta_y,\eta_z )=(\nu_\Id,\nu_x,\nu_y,\nu_z ) =( -1,0,1,1)$ are a degree-4 $\SLIP$ measure for odd and even number of qubits respectively \cite{PhysRevA.85.022301}.

A general method to construct various $\SLIP$ measures based on this simple idea of exploring the anti-linear expectation value of the tensor product of Pauli operators was further invented in Ref~\cite{PhysRevA.72.012337,Eltschka_2014,PhysRevA.85.022301}. In particular

Any SLIP measure $E$ can also be extended to mixed states by determining the largest convex function on the set of mixed states which coincides with $E$ on the set of pure states  \cite{Uhlmann98}. Despite its simple definition, the evaluation of a convex roof extension requires non-linear minimization procedure, and for a general density matrix is a challenging task \cite{PhysRevLett.114.160501,Osborne_2006,Regula_2014,Regula_2016a}. An attempt to address this challenging task was carried out by introducing the so-called \textit{zero-polytope}, the convex hull of pure states with vanishing $E$ measure \cite{OsterlohTangles,OsterlohWernerStates,
OsterlohExactZeros,Osterloh3}. In the particular case of rank-2 density matrices $\rho $, the zero-polytope can be represented inside a Bloch sphere, spanned by the roots of $E$ \cite{OsterlohWernerStates,RegulaGeoTanglePRL}. We adapt this approach focusing only on the roots of polynomial invariants, equivalently the vertices of the zero-polytope.

\section{System of roots}
\label{sec:System of roots}

Consider a $(n+1)$-partite qubit state $\ket{\psi}$.
The state  $\ket{\psi}$ can be uniquelly written as
\begin{equation}
\label{Peq1}
\ket{\psi}=\ket{0} \ket{\psi_0} + \ket{1}\ket{\psi_1}\,,
\end{equation} 
providing the canonical decomposition of the reduced density matrix $\rho =\ket{\psi_0}\bra{\psi_0} +\ket{\psi_1}\bra{\psi_1}$ obtained by tracing out the first qubit.
Note that the states $\ket{\psi_0}$ and $\ket{\psi_1}$ are in general neither normalized nor orthogonal. Consider now the family of states
\begin{equation}
\label{psi0psi1}
\ket{\psi_z} = z \ket{\psi_0} +  \ket{\psi_1},
\end{equation}
where $z$ is taken from the extended complex plane $\hat{\mathbb{C}}$, i.e., complex numbers plus infinity. We denote this the \textit{extended plane representation}. In addition, consider any $\SLIP_n^h$ measure $E$ defined on the set of $n$-partite pure qubit states.
Since $E$ is polynomial in the coefficients of $\ket{\psi_z}$, it is also polynomial in the complex variable $z$  \cite{OsterlohTangles}.
Therefore, the polynomial $E( z\ket{\psi_0} + \ket{\psi_1} )$ has exactly $h$ roots: $\zeta_1 ,\ldots ,\zeta_h $ (which may be degenerated and/or at infinity), related to the degree of $E$. By using the complex number $z$, the states $\ket{\psi_z}$ can be mapped to the surface of a sphere via the standard \textit{stereographic projection}
$(\theta,\phi ):=
(\text{arctan} \,1/ |z|,\; -\text{arg}\,z )$
written in spherical coordinates.
This way, a point on the unit 2-sphere $(\theta,\phi)$ can be associated with the quantum state
\begin{equation}
\ket{\widetilde{\psi}_{z}} :=
 \text{cos} \dfrac{\theta}{2} \ket{\psi_0} +
 \text{sin} \dfrac{\theta}{2} e^{i \phi} \ket{\psi_1}
\end{equation}
with $z =\text{ctg}({\theta}/{2}) \: e^{-i \phi}$, such that $\ket{\psi_0}$ lies in the North pole and $\ket{\psi_1}$ lies in the South pole, see \cref{Fig:StereoProjection}. We denote this the \textit{Bloch sphere representation}. Note that ${\ket{\widetilde{\psi}_{z}} \propto \ket{\psi_{z}}}$ and that neither of these states is normalized, since $\ket{\psi_0}$ and $\ket{\psi_1}$ are not normalized in general either.

\begin{figure}[h!]
\includegraphics[width=.9\columnwidth]{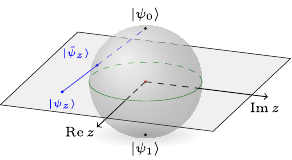}
\caption{The stereographic projection relating the family of states $\ket{\psi_z}$ on the extended complex plane with the associated family of states $\ket{\widetilde{\psi_z}}$ on the Bloch sphere. The spherical coordinates $({\theta},\phi)$ and the complex coordinate $z$ are related via the stereographic projection $z =\text{ctg} ({\theta}/{2}) \: e^{-i \phi}$.
}
\label{Fig:StereoProjection}
\end{figure}

\subsection{Local operations on the system of roots}
\label{sec:Local operations on the system of roots}

To each linear invertible operator $\mathcal{O} =\begin{psmallmatrix}
a&b\\
c&d
\end{psmallmatrix}$, one may associate a \textit{Möbius transformation} $z\mapsto z':= \frac{az+b}{cz+d}$, mapping the extended complex plane $\hat{\mathbb{C}}$ into itself \cite{bengtsson_zyczkowski_2006,doi:10.1142/S0219749912300045}.
The composition of such transformations represents the multiplication of the associated operators.
In particular, $z\mapsto z':= \frac{dz-b}{-cz+a}$ is an inverse Möbius transformation related with
$\mathcal{O}^{-1} =
\begin{psmallmatrix}
d&-b\\
-c&a
\end{psmallmatrix}$.
Note that although Möbius transformations are typically represented on the extended complex plane, one may represent them as transformations on the Bloch sphere via the stereographic projection.
The correspondence between invertible operators and Möbius transformations represented on the Bloch sphere was already successfully used for SLOCC classification of permutation-symmetric states \cite{Bastin_2009,CrossRatioNqubits,Three-tangleMajorana}.

To study the effect of SLOCC operations on the system of roots we begin by acting on the first qubit of a state $\ket{\psi}$ written in the form of Eq.~(\ref{Peq1}) with an invertible linear operator
$\mathcal{O}$. In terms of the family of states $\ket{\psi_{z}}$ in Eq.~(\ref{psi0psi1}), this operation induces the map
\begin{equation}
\label{psiMob}
\ket{\psi_{z}} \mapsto \ket{\psi_{z'}} = \frac{az+b}{cz+d} \ket{\psi_0} + \ket{\psi_1}\,,
\end{equation}
i.e. the index is mapped via the Möbius transformation $z\mapsto z':= \frac{az+b}{cz+d}$, see \cref{TheoremProof}. In addition, since ${\ket{\widetilde{\psi}_{z}} \propto \ket{\psi_{z}}}$, we have that the family of states $\ket{\widetilde{\psi}_{z}}$ also transforms according to Eq.~(\ref{psiMob}). This reflects the fact that the states $\ket{\psi_z}$ and $\ket{\widetilde{\psi}_z}$ associated to the extended complex plane and the Bloch sphere are related by a stereographic projection of the variable $z$. Using Eq.~(\ref{psiMob}) and the defining equation $E (\zeta_i \ket{\psi_0}+ \ket{\psi_1})=0$ for the roots $\zeta_i$ of the polynomial $E$, one concludes that the roots transform according to the inverse Möbius
transformation associated to the operator $\ma{O}$, i.e. $\zeta_i \mapsto  \frac{d\zeta_i-b}{-c\zeta_i+a}$.
Finally, although the system of roots changes with local operations acting on the qubit that is being traced out in Eq.~(\ref{Peq1}), it is invariant under local operations acting on any other qubit since the polynomial $E$ is $\SLC^{\otimes n}$ invariant.
We summarize these results in the theorem below, see \cref{TheoremProof} for a detailed proof.

\begin{theorem}
\label{T1}
Consider an $(n+1)$-partite pure quantum state $\ket{\psi} =\ket{0}\ket{\psi_0}+\ket{1}\ket{\psi_1}$.
The roots $\zeta_i$ of any $\SLIP_n^h$ entanglement measure associated to the partial trace of the first qubit:
\begin{enumerate}
\item are invariant under invertible operators, i.e. invariant under $\textbf{1}\otimes\mathcal{O}_{\vec{n}} \in \SLC^{\otimes n}$ operators;
\item transform via an inverse Möbius transformation $\zeta_i '=\frac{d\zeta_i-b}{-c\zeta_i+a}$ w.r.t the
$\mathcal{O} =\begin{psmallmatrix}
a&b\\
c&d
 \end{psmallmatrix}\otimes \textbf{1}^n$ operator.
\end{enumerate}
\end{theorem}

\noindent
It is crucial to emphasize that normalizing the states $\ket{\psi_0}$ and $\ket{\psi_1}$ after the action of the operator $\ma{O}$, as is the case in existing related works \cite{OsterlohTangles,OsterlohWernerStates,Osterloh3,RegulaGeoTanglePRL,RegulaGeoTanglePRA},
would spoil the mapping of \cref{psiMob}.
As a consequence, the action of SLOCC operators on the states $\ket{\psi_z}$ would no longer be given by the corresponding Möbius transformation, and the statements in \cref{T1} would no longer hold.

The decomposition (\ref{Peq1}) can be performed with respect to any other subsystem, each with its own system of roots. Any local operator $\mathcal{O}_k
=\begin{psmallmatrix}
a&b\\
c&d
\end{psmallmatrix}$ acting on the $k$-th qubit will influence independently the corresponding $k$-th system of roots via the Möbius transformation $\zeta_i \mapsto  \frac{d\zeta_i-b}{-c\zeta_i+a}$. On the other hand, if acting globally with a local operator $\mathcal{O}_1 \otimes \cdots \otimes \mathcal{O}_{n+1}$, all roots (and thus all zero-polytopes) will be affected.
Since a Möbius transformation is a bijective mapping on the Bloch sphere, the total number of roots will always be preserved \cite{RegulaGeoTanglePRA}.
Moreover, the existence of a local transformation between two given states becomes straightforward to verify since Möbius transformations are fully classified.

\begin{figure*}
\includegraphics[scale=1.6]{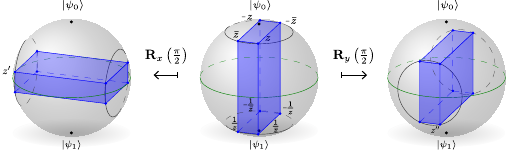}
\caption{A normal system of roots $z,1/z,-z,-1/z$ together with the conjugate points $\bar{z},1/\bar{z},-\bar{z},-1/\bar{z}$ span the cuboid whose faces are parallel to the $XZ$, $XY$ and $YZ$ planes.
There are 24 rotations of the Bloch sphere which preserve this property, composing the elements of the group $\ma{G}_{24}$.
Two of them, namely the rotation by a $\pi /2$ angle around $X$ and $Y$ axes are presented.
The system of roots transforms according to Eqs.~\ref{Indeed}-\ref{Indeed3}, giving
$z\mapsto z':= \frac{z-i}{-iz+1}$ and $z\mapsto z'':= \frac{z-1}{z+1}$ for the two rotations.
}
\label{G24}
\end{figure*}

\section{Verification of SLOCC-equivalence}
\label{sec:Verification of SLOCC-equivalence}

\cref{T1} provides a solution for the problem of discriminating quantum states up to SLOCC-equivalence. To verify if two pure states are SLOCC-equivalent, one can use the following procedure, which takes a single $\SLIP_n^h$ measure and two pure $(n+1)$-qubit states as an input. In the generic situation, it returns a set of at most $(3! {{h}\choose{3}})^{n+1}$ SLOCC operators as an output. The input states are SLOCC-equivalent iff they are interconnected by one of the output operators.

\begin{procedure}
\label{procedure}
Choose any $\SLIP_n^h$ entanglement measure of degree $h \geq 3$ and two $(n+1)$-qubit states.
\begin{itemize}
\item[\textbf{1)}] Calculate the roots of a chosen measure for each subsystem for both states. If for any subsystem the number of roots of both states is not equal, such states are not SLOCC-equivalent. Otherwise, denote by $h_i$ the number of roots of both states calculated for $i$-th subsystem. If  $h_i \geq 3$ for all $1\leq i\leq n+1$, the procedure will be conclusive. Note that for a generic state $h_i=h\geq 3$ for each subsystem $i$.
\item[\textbf{2)}] Focus on one subsystem $i$, $1\leq i\leq n+1$, and determine all Möbius transformations which transforms the roots of the first state into roots of the second state. For example, choose three of the $h_i$ roots from each state and write the unique Möbius transformations between the two triplets of roots. Repeat this for all $3! {{h_i}\choose{3}}$ possibilities of choosing $3$ out of $h_i$ roots for the second state. Derive the local operators $\ma{O}_i$ associated to Möbius transformations.
\item[\textbf{3)}] Repeat step 2) for all other subsystems and then consider the tensor products $\ma{O}_1 \otimes \dots \otimes \ma{O}_{n+1}$ of all the local operators obtained.
\item[\textbf{4)}] If the two given $(n+1)$-qubit states are SLOCC-equivalent, at least one of these operators must transform one state into the other (up to the normalization). Otherwise, they are not SLOCC-equivalent.
\end{itemize}
\end{procedure}



\begin{proposition}
\label{proppppp1}
Consider any $\SLIP_n^h$ measure and two $(n+1)$-qubit states. If both states have at least $3$ roots with respect to each subsystem, they are SLOCC-equivalent iff they are interconnected by one of the operator obtained as an outcome of \cref{procedure}.
\end{proposition}

\noindent
We refer to \cref{22november} for a proof of this statement. 
Intuitively, a generic pure quantum state shall have exactly $h$ distinct roots for each subsystem, where $h$ is the degree of the $\SLIP_n^h$ measure. Thus, any such measure of degree $h\leq 3$ should be sufficient to verify SLOCC-equivalence between generic quantum states. We confirmed this intuition by looking at several examples of degree-$4$ measures. In particular, \cref{(2k),(2k+1)} provides a family of $\SLIP$ measures for arbitrary number of qubits. 
We examined those measures, and for each number $4 \leq n \leq 20$ of qubits, we generated a sample of a thousand $(n+1)$-qubit states and each state had, indeed, four distinct roots with respect to any subsystem \cite{commentSLIPgenerate}. Furthermore, we confirm this numerical result analytically and proved that a generic pure quantum state has, indeed, four distinct roots for each subsystem for measures defined via \cref{(2k),(2k+1)}. By generic state, we mean the set of all pure states except of the measure zero subset with respect to the Haar measure, see \cref{22Jun} for more details. This result applies for an arbitrary number of qubits, thus proving the following proposition.

\begin{proposition}
\label{easyBusyMainText}
A single SLIP measure is enough to provide necessary and sufficient conditions for any two generic pure $n$-qubit states to be SLOCC-equivalent.
\end{proposition}

\noindent
We refer to \cref{Section 5} for a technical proof of the above statement. 



\section{Normal system of roots}
\label{sec:Normal system of roots}

In the previous section we showed that in principle any SLIP measure of degree  $h\geq 3$ can be used to verify SLOCC-equivalence. In this section, we consider the special case when such a measure has degree $h=4$.  

Firstly, recall that any three distinct points on the sphere can be transformed onto any other three distinct points via a unique Möbius transformation, see \cref{11november}. While this is not the case for four points, it is possible to take any four complex points $z_1,z_2,z_3, z_4$ and associate a so-called \textit{cross-ratio}
\begin{equation}
\label{CRR}
\lambda
\big(z_1,z_2,z_3,z_4 \big):=
\dfrac{z_3 -z_1}{z_3 -z_2}\dfrac{z_4 -z_2}{z_4 -z_1}\,,
\end{equation}
which is preserved under Möbius transformations \cite{CrossRatioNqubits,bengtsson_zyczkowski_2006}. Systems of four distinct points are related via Möbius transformations if their cross-ratios are related in the same way. The cross-ratio is not invariant under permutations of points, however, and depending on the ordering taken for the four points, it takes six values:
$\lambda, \frac{1}{\lambda}, 1-\lambda ,\frac{1}{1-\lambda},\frac{\lambda-1}{\lambda},\frac{\lambda}{\lambda-1}$
\cite{CrossRatioNqubits}. A particular interesting set of four points is one of the form $z,1/z, -z,-1/z $, which we call a \textit{normal system}.
Any set of four points may be mapped into a normal system, for which $z,1/z, -z,-1/z $ will be the solutions of the fourth degree equation $\lambda=4z^2/(1+z^2)^2$, where $\lambda$ is the corresponding cross-ratio from Eq.~(\ref{CRR}). Such a map is unique up to symmetries of the cube, i.e the group of $24$ rotations generated by $\pi/2$ rotations along the $X,Y,Z$ axis, denoted by $\mathcal{G}_{24}$. \cref{NorFor} contains the exact description of the group $\mathcal{G}_{24}$. In \cref{secNorForm} we formally show the following mathematical result.

\begin{proposition}
\label{24}
Each non-degenerated four points $z_1, z_2, z_3, z_4$ on the Bloch sphere can be transformed onto the normal form $z,\frac{1}{z},-z,-\frac{1}{z} $ via a Möbius transformation $T$.
The latter is uniquely defined up to $24$ rotations in the group $\mathcal{G}_{24}$.
\end{proposition}

This mathematical result has interesting implications for the SLOCC verification and classification problems. Indeed, assume that $E$ is $\SLIP^4_n$ measure of degree $h=4$, and the state $\ket{\psi} \in \mathcal{H}_2^{\otimes (n+1)}$ has exactly four distinct roots for each subsystem. By \cref{24}, one can find a Möbius transformation $T_i$ that transforms roots for each system $i$ into the normal form. By \cref{T1}, the local operators $\mathcal{O}_i$ corresponding to transformations $T_i$ transform state $\ket{\psi}$ into the form with roots for $i$-th system being in normal form, hence
\begin{equation}
\ket{\psi'}:=\mathcal{O}_1\otimes\cdots\otimes\mathcal{O}_{n+1}\ket{\psi} 
\end{equation}
is a state for which roots of measure $E$ for each subsystem are in normal form. We call $\ket{\psi'}$ a $E$-\emph{normal form} of a state $\ket{\psi}$ with respect to measure $E$. Note that the $E$-normal form of a state, if it exists, is defined up to the group $\mathcal{G}_{24}^{n+1}$ of local rotations. 
We illustrate this procedure on a simple example of broadly discussed four-partite state $\frac{1}{\sqrt{2}}(\ket{0}\ket{\textrm{GHZ}}+\ket{1}\ket{\textrm{W}})$ \cite{Osterloh3,OsterlohWernerStates}. We calculate its roots for $\tau^{(3)}$ measure (\ref{3tangle}) and find the Möbius transformation transforming them into a normal system, see \cref{NorForEx} for detailed calculations. As a result, we obtain SLOCC operator that transforms the aforementioned state into its normal form with respect to  $\tau^{(3)}$ measure, see \cref{MobiusTrans}.

In this way, for the chosen $\SLIP^4_n$ measure $E$, we defined the $E$-normal form of any state that has exactly four roots with respect to measure $E$ for each system. As we have shown in \cref{easyBusyMainText}, choosing for example measure $E$ defined on \cref{(2k),(2k+1)}, provides $E$-normal form of the generic state of any number of qubits. 

This allows us to address more ambitious problems, when for example, one needs to verify the pairwise SLOCC-equivalence of a larger number of states. Indeed, one can reduce each of them to $E$-normal form and then compare them using only a finite (and relatively small) group $\mathcal{G}_{24}^{n+1}$ of local rotations. 
Furthermore, such $E$-normal form might be possibly used for an even more challenging task of SLOCC-classification of pure states. Indeed, as we will demonstrate in the next section, the $\tau^{(3)} $-normal form for measure (\ref{3tangle}) can be successfully used to classify generic four qubit states. Recall that the problem of classifying $n\leq 5$ states still remains open \cite{PolInv4qubits,Zhang_2016,KempfNessToEntanglement,BurchardtRaissi20}.

We shall finish this section by pointing out the intriguing connection between $E$-normal form of a state and, so-called, \textit{normal form} of a state \cite{NormalFormVerstraete}. Recall that a state is in a normal form if reduced density matrices to one subsystem are all maximally mixed, i.e. proportional to the identity. Normal form, if it exists, is defined up to the local unitary operations. The process of determining the normal form of a state, if it exists, may turn to be an infinite iterative minimization process \cite{NormalFormVerstraete}. On the other hand, the process of obtaining $E$-\emph{normal form} of a state is straightforward and consists of a finite number of steps. As we observed, for four qubit states the normal form of a state coincides with its $\tau^{(3)}$-normal form, see \cref{NorForEx}. For systems with a larger number of qubits $n>4$ (especially for $n=5.6$), we unsuccessfully searched for measures $E$ for which $E$-normal form would coincide with normal form of a state. Such measures, if found, would lead to a simple procedure for obtaining a normal form of a state for arbitrary number of qubits.


We illustrate this procedure by transforming the widely discussed four-partite state $\frac{1}{\sqrt{2}}(\ket{0}\ket{\textrm{GHZ}}+\ket{1}\ket{\textrm{W}})$ \cite{Osterloh3,OsterlohWernerStates} into its normal form, see \cref{MobiusTrans}. Without this technique, the standard way of obtaining the normal form would indeed result in an infinite iterative procedure.


\begin{figure}[h!]
\includegraphics[width=.99\columnwidth]{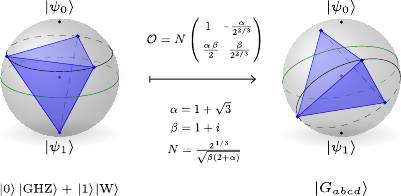}
\caption{The system of four roots (represented as blue dots) related to the 3-tangle polynomial measure $\tau^{(3)}$ evaluated on the first subsystem of the state $1/\sqrt{2}(\ket{0}\ket{\textrm{GHZ}}+\ket{1}\ket{\textrm{W}})$.
This system of four points can be mapped into a normal system (i.e. symmetrically related points $z,-z,1/z,-1/z$) by a Möbius transformation. Similar local transformations can be performed with respect to other subsystems, transforming the states into a state in the normal form.
}
\label{MobiusTrans}
\end{figure}

\section{State classification}
\label{sec:State classification}

We show that for small numbers of qubits $n=3,4$, our approach might be successfully used for the more demanding problem of  entanglement classification. 


Focusing first on the three-qubit case, genuinely entangled pure states are SLOCC-equivalent to either $\ket{\textrm{GHZ}}=\frac{1}{\sqrt{2}}(\ket{000}+\ket{111})$ or $\ket{\textrm{W}}=\frac{1}{\sqrt{3}}( \ket{001}+\ket{010}+\ket{100})$ \cite{threeQubits}. Using the 2-tangle $\tau^{(2)}$ \cite{PhysRevLett.80.2245} as the entanglement measure, one may use the roots to distinguish between the two classes. Indeed, all rank-2 reduced density matrices of the $\ket{\textrm{W}}$ state have a single root, while there are always two distinct roots for the $\ket{\textrm{GHZ}}$ state \cite{RegulaGeoTanglePRL}.

Contrary to the three qubit case, there are infinitely many SLOCC classes of four qubit states \cite{threeQubits}.
Although four qubit states were divided into nine families \cite{FourQubits,Djokovic4qubits,SpeeKraus}, we will focus on generic 4-qubit states, i.e. 4-qubit states with random coefficients belonging to the so called $G_{abcd}$ family - the 4-qubit SLOCC family with the most degrees of freedom. The representative state is of the form $\ket{G_{abcd}}=\frac{a+d}{2} \big(\ket{0000}+\ket{1111} \big)+\frac{a-d}{2} \big(\ket{0011}+\ket{1100} \big)+\frac{b+c}{2}  \big(\ket{0101}+\ket{1010} \big)+\frac{b-c}{2} \big(\ket{0110}+\ket{1001} \big)$,
where $a^2 \neq b^2 \neq c^2 \neq d^2 $ are pairwise different. Choosing the 3-tangle $\tau^{(3)}$ (\ref{3tangle}) as the entanglement measure, the states $\ket{G_{abcd}}$ have four non-degenerate roots already in the normal form, see \cref{GabcdProof}. Since the normal form of roots is unique up to the group $\mathcal{G}_{24}$, the problem of SLOCC-equivalence of states $\ket{G_{abcd}}$ becomes solvable, with a discrete amount of solutions. Indeed, it can quickly be confirmed if two states in the $G_{abcd}$ class are SLOCC equivalent by checking if one can be obtained from the other by the action of an element of the finite class of operators $\mathcal{O}\in \mathcal{G}_{24}^{\otimes 4}$. We thus find that exactly 192 states of the form $\ket{G_{a bcd}}$ are SLOCC-equivalent.

\begin{proposition}
Two states $\ket{G_{a bcd}}$ and $\ket{G_{a' b'c'd'}}$ are SLOCC-equivalent iff their coefficients are related by the following three operations: multiplication by a phase factor $(a', b',c',d')=e^{i\phi} (a ,b,c,d)$, and permutation of coefficients $(a', b',c',d')=\sigma (a ,b,c,d)$, and change of sign in front of two coefficients from $a,b,c,d$.
\label{propGabcd}
\end{proposition}

\noindent
See \cref{GabcdProof} for a detailed proof of the above statement. We finish this section by pointing out some intriguing connections between this result and other related problems. Note that the symmetry in \cref{propGabcd} is given by the Weyl group of Cartan type $D_4$ and it has already been observed that the generators of four-qubit polynomial invariants exhibit this type of symmetry \cite{PolynomialInvariantsofFourQubits,Djokovic4qubits}. As a consequence, this result constitutes a new relation between 4-qubit invariants and the convex roof extension of 3-tangle $\tau^{(3)}$, which may shed some light on the problem of generalizing the CKW inequality \cite{DistributedEntanglement} for four qubit states \cite{GourWallach,ProblemTangle,Regula_2014,Eltschka_2014,Regula_2016a},
and beyond
\cite{ProblemTangle,GourWallach,MonogamyEqualitiesDeg2and4}.


\section{Conclusions}
\label{sec:Conclusions}

In this paper, we showed how a single $\SLIP_n^h$ entanglement measure is enough to verify if two generic $(n+1)$-qubit states to be SLOCC equivalent. Our result is applicable for any number of qubits. This was possible by showing that the roots of any $\SLIP_n^h$ measure transform via Möbius transformations under the SLOCC operations performed on the subsystems. In that way, SLOCC equivalence between two states is implied by the easily verifiable existence of a Möbius transformation relating aforementioned roots for each subsystem. Moreover, we define the $E$-normal form of the state with respect to the given $\SLIP_n^{h=4}$ measure $E$, for which roots are symmetrically distributed on the Bloch sphere. 
Such form is simple to determine and exists for generic states. In comparison, so-called, normal form of state requires possibly infinite iterative minimization procedure. 
Furthermore, we demonstrated our approach on 4-qubit states, and showed that the roots of the 3-tangle measure $\tau^{(3)}$ are enough to fully classify 4-qubit states from the most generic $G_{abcd}$ family. Lastly, as we observed that for 4-qubit states the $\tau^{(3)}$-normal form coincides with the normal form of a state, which gives a procedure to determine the normal form of state that circumvents the possibility of an infinite iterative process of the standard procedure.



\section*{Acknowledgements}

The authors thank Andreas Osterloh, Karol \.{Z}yczkowski, Rui Perdig\~{a}o and Yasser Omar for fruitful discussions and correspondence. The authors are especially grateful to Jens Siewert, Andreas Osterloh, Barbara Kraus, and Karol \.{Z}yczkowski for stimulating conversations which led to formal proof of Proposition 2 while staying at Centro de Ciencias de Benasque Pedro Pascual. A.B acknowledges support from the National Science Center under grant number DEC-2015/18/A ST2/00274 and by NWO Vidi grant (Project No. VI.Vidi.192.109). G.Q. thanks the support from Funda\c{c}\~{a}o para a Ci\^{e}ncia e a Tecnologia (Portugal), namely through projects CEECIND/02474/2018 and project UIDB/50008/2020 and IT project QuantSat-PT. R.A. acknowledges support from the Doctoral Programme in the Physics and Mathematics of Information (DP-PMI) and the Funda\c c\~ao para a Ci\^encia e Tecnologia (FCT)  through Grant No.~PD/BD/135011/2017. 
This work is funded/co-funded by the European Union (ERC, ASC-Q,
101040624). Views and opinions expressed are however those of the authors
only and do not necessarily reflect those of the European Union or the
European Research Council. Neither the European Union nor the granting
authority can be held responsible for them.





\section*{Appendices}
In appendices, we present detailed proofs of the statements presented in the main body of the paper, as well as a summary of the related results concerning Möbius transformations, SLIP entanglement measures, Haar measure on the set of pure quantum states, and SLOCC-classification. 

\cref{TheoremProof,22november} contain proofs of Theorem 1 and Proposition 1 respectively. 

\cref{22Jun} summarizes basic results regarding the Haar measure on the set of pure quantum states and discusses the notion of generic quantum states. \cref{Section 5} proves the main result of a paper, namely \cref{easyBusyMainText}. The proof is based on two technical lemmas related to the systems of odd and even number of qubits, see \cref{lemma1odd} and \cref{lemma1even} respectively.

\cref{11november} presents a well-known form of a unique Möbius transformation transforming two given tuples of 3 points one onto another. \cref{NorFor,secNorForm} discuss details regarding an $E$-normal form of a state with respect to a given measure $E$. Furthermore, \cref{NorForEx} provides an explicit example of transformation of a state into its $E$-normal form. 

\cref{GabcdProof} presents proof of \cref{propGabcd}, which is reproducing the SLOCC-classification results for the generic systems of four qubits. The $E$-normal form proves to be useful for that purpose. 

Lastly, in \cref{24September}, we apply the presented method to verify the SLOCC equivalence between certain $5$-qubit states.

\appendix

\section{Proof of \cref{T1}}
\label{TheoremProof}

We present a proof for Theorem 1. Any $n+1$ partite qubit state $\ket{\psi} \in \mathcal{H}_2^{\otimes (n+1)}$ might be written as
\begin{equation}
\label{Peq111}
\ket{\psi}=\ket{0} \ket{\psi_0} + \ket{1}\ket{\psi_1}.
\end{equation}
Such a form provides the canonical decomposition of the reduced density matrix $\rho_{} =\ket{\psi_0}\bra{\psi_0} +\ket{\psi_1}\bra{\psi_1}$ over the non-normalized states $\ket{\psi_0}, \ket{\psi_1}\in \mathcal{H}_2^{\otimes N}$, obtained by tracing out the first qubit. Consider now a reversible operator
$\mathcal{O} =\begin{psmallmatrix}
a&b\\
c&d
\end{psmallmatrix} \in \SLC$ acting on the first qubit. Under the action of this operator, the state $\ket{\psi}$ is transformed into
\begin{align*}
\ket{\psi'}&:=
\mathcal{O} \ket{\psi}= \Big(a\ket{0}+b\ket{1}\Big)\ket{\psi_0} \\&+ \Big(c\ket{0}+d\ket{1}\Big) \ket{\psi_1}
= \ket{0}\ket{\psi'_0} + \ket{1}\ket{\psi'_1}\,,
\end{align*}
where
\ea{
\ket{\psi'_0} & := a\ket{\psi_0}+c\ket{\psi_1} \label{p0l}\,, \\
\ket{\psi'_1} & := b\ket{\psi_0}+d\ket{\psi_1} \label{p1l}\,.
}
Consider now any superposition of states $\ket{\psi'_0}$ and $\ket{\psi'_1}$. Observe that
\begin{align*}
\ket{\psi'_z} := z \ket{\psi'_0}+\ket{\psi'_1}
&=z \Big(a\ket{\psi_0} +c \ket{\psi_1} \Big)  \\
&+  b\ket{\psi_0} +d \ket{\psi_1} \\
&=\left(a z+b\right) \ket{\psi_0} +\left( cz+d \right) \ket{\psi_1}\\
&\propto \dfrac{az+b}{cz+d}\ket{\psi_0 }+ \ket{\psi_1},
\end{align*}
where the compex number $cz+d$ was factored out in order for the transformation to map states from the extended plane representation to the extended plane representation.  In other words, we have
\eq{
\ma{O}\ket{\psi_z} = \ket{\psi_{z'}}\,, \quad z' = \frac{az+b}{cz+d}\,,
}
i.e. the operator $\ma{O}$ transforms states in the extended plane representation by applying a Möbius transformation on the index $z$. Suppose now that $\zeta_i$ is a zero of a $h$-degree polynomial function $E$, i.e. $E (\zeta_i \ket{\psi_0}+ \ket{\psi_1})=0$. Acting on the first qubit with $\ma{O}$, the density matrix after tracing out the first qubit becomes $\ket{\psi'_0}\bra{\psi'_0} +\ket{\psi'_1}\bra{\psi'_1}$, so the entanglement measure $E$ will be zero for some new roots $\zeta'_i$, such that $E (\zeta'_i \ket{\psi'_0}+ \ket{\psi'_1})=0$. Using Eqs.~(\ref{p0l})-(\ref{p1l}), the later equation can be transformed into
\eq{
E\left((c\zeta_i'+d)\left(\dfrac{a\zeta_i'+b}{c\zeta_i'+d} \ket{\psi_0}+ \ket{\psi_1}\right) \right)=0
}
where the factor $(c\zeta_i'+d)$ is irrelevant since any root multiplied by it will still be a valid root. Comparing with the equation for the roots before the action of $\ma{O}$, we reach the conclusion that the roots transform according to the inverse Möbius transformation as
\eq{\label{Mzeta}
\zeta'_i =\dfrac{d\zeta_i-b}{-c\zeta_i+a} \,,
}
under the action of the operator $\ma{O}$. As a consequence, the roots of the zero-polytope transform with respect to the inverse Möbius transformation associated to the operator $\mathcal{O} =\begin{psmallmatrix}
a&b\\
c&d
\end{psmallmatrix}$.
Analize now the case when $\mathcal{O}$ is a unitary operator $\mathcal{O}=\mathcal{U}$.
Since any unitary operator $\mathcal{U}$ can be represented as a rotation $\mathcal{R} =\begin{psmallmatrix}
\text{cos}\alpha & \text{sin}\alpha \;e^{i \phi}\\
-\text{sin}\alpha \;e^{-i \phi} & \text{cos}\alpha \end{psmallmatrix}$
(up to an irrelevant global phase), it will simply rotate the Bloch ball, together with the zero-polytope.

Consider now multilocal operators $\mathcal{O}_{\vec{n}}= \mathcal{O}_1\otimes\ldots\otimes\mathcal{O}_n$ acting on the remaining qubits of the state $\ket{\psi}$ from \cref{Peq111}.
The state $\ket{\psi}$ will transform accordingly as
\eq{
\ket{\psi'}:=
\mathcal{O}_{\vec{n}}  \ket{\psi}=\ket{0}
\underbrace{\mathcal{O}_{\vec{n}} \ket{\psi_0} }_{:=\ket{\psi_0'}}+
\ket{1}
\underbrace{\mathcal{O}_{\vec{n}} \ket{\psi_1} }_{:=\ket{\psi_1'}}.
}
After the action of $\mathcal{O}_{\vec{n}}$, a state in the extended plane representation will have a value of entanglement measure $E$ equal to
\[
E \Big( z\ket{\psi_0'}+ \ket{\psi_1'} \Big)=
E \Big( \mathcal{O}_{\vec{n}} \big(z\ket{\psi_0}+ \ket{\psi_1}\big) \Big)\,.
\]
However, since $E$ is $\SLC^{\otimes n}$ invariant, we have that $E (z \ket{\psi_0}+ \ket{\psi_1} )=0$ iff $E (z \ket{\psi_0'}+ \ket{\psi_1'} )=0$, and so the roots of both polynomial equations are the same. As a consequence, the roots of the the zero-polytope will remain unchanged under the action of $\mathcal{O}_{\vec{n}}$. This concludes the proof of Theorem~1.

\section{Proof of \cref{proppppp1}}
\label{22november}

In order to prove \cref{proppppp1}, we begin with the following simple observation:

\begin{observation}
For two sets $\Lambda =\{z_1,\ldots,z_h\}$ and $\Lambda' =\{w_1,\ldots,w_h\}$ of $h$ distinct  elements of the extended complex plane $z_i,w_i \in \hat{\mathbb{C}}$, $h>2$, there exists at most $h(h-1)(h-2) =3! {{h}\choose{3}}$ Möbius transformations which maps $\Lambda$ into $\Lambda'$.
\end{observation}

\begin{proof}
Suppose that $T$ is a Möbius transformation which transforms $\Lambda$ into $\Lambda'$. Thus,
\[
T(z_1)=w_{i_1} ,\quad T(z_2)=w_{i_2}, \quad T(z_3)=w_{i_3},
\]
for distinct indices $i_1,i_2,i_3 =1,\ldots,h$. Observe that there are exactly $h(h-1)(h-2)$ possibilities for choice of those indices. For each such possibility, there is an unique Möbius transformation, which maps
$z_1 \mapsto w_{i_1}, z_2 \mapsto w_{i_2} ,z_3 \mapsto w_{i_3} $, see \cref{Section 5} \cite{TothMTbook}. Therefore there are at most $h(h-1)(h-2)$ Möbius transformations which maps $\Lambda$ into $\Lambda'$.
\end{proof}

\begin{remark}
\label{1Dec}
Usually, there are more efficient methods to determine all Möbius transformations which maps a given set $\Lambda$ into the other $\Lambda'$. For example, if $|\Lambda|=|\Lambda'|=4$, one may calculate the values of cross-ratios. If they differ for any choice of orderings in $\Lambda$ and $\Lambda'$, those sets are not related by any Möbius transformation.
\end{remark}

\begin{proof}[Proof of Proposition 1]
Suppose that two $(n+1)$-qubit states $\ket{\psi},\ket{\phi} \in \mathcal{H}_2^{\otimes (n+1)}$ have at least $3$ roots calculated with respect to a given $\SLIP_n^h$ measure $E$ for each subsystem. We shall show that both states are SLOCC-equivalent iff in the outcome of Procedure 1 there is an operator $\mathcal{O}$ providing such equivalence.

Firstly, notice that if in the outcome of Procedure 1 there is an operator providing SLOCC-equivalence, both states are SLOCC-equivalent. Hence, described condition is a sufficient condition for SLOCC-equivalence.

Secondly, we shall see that the described condition is a necessary condition for SLOCC-equivalence. Indeed, suppose that states $\ket{\psi},\ket{\phi}$ are SLOCC-equivalent, i.e. there exists an operator
\begin{equation}
\label{Dec1prim}
\mathcal{O} =\mathcal{O}_1 \otimes \cdots\otimes \mathcal{O}_{n+1}
\end{equation}
such that $\mathcal{O}\ket{\psi} \propto \ket{\phi}$ for some operators $\mathcal{O}_i \in  \SLC$. Denote by $T_i$ the Möbius transformation corresponding to $\mathcal{O}_i$. Denote by $\Lambda_i =\{z_1^i,\ldots,z_{h_i}^i\}$ the system of roots of the state $\ket{\psi}$ calculated with respect to $i$-th subsystem. According to Theorem 1, the $i$-th system of roots of a state $\ket{\phi}\propto  \mathcal{O}\ket{\psi} $ is simply  $\Lambda_i ' =\{T_i(z_1^i),\ldots,T_i(z_{h_i}^i)\}$. In particular $T_i$ maps $\Lambda_i$ into $\Lambda_i'$, hence will be found in Step \textbf{2} of Procedure 1. Note, that there might be other operators transforming $\Lambda_i$ into $\Lambda_i'$, hence in Step \textbf{2}, one may obtain multiple local operators. Similarly for any other subsystem, the transformation $T_i$ (and the corresponding local operator $\mathcal{O}_i$) will be found. Hence among all local operators in the outcome of Procedure 1, there will be operator $\mathcal{O}$ from \cref{Dec1prim}.
\end{proof}


\section{Fubini–Study measure and generic states}
\label{22Jun}

A \textit{Fubini–Study measure} of the set of pure quantum states (known also as a Haar measure) is a unique probability distribution on $\mathcal{H}_{d+1}$ which is invariant under any unitary operation acting on the whole $\mathcal{H}_{d+1}$. It can be explicitly written as a probability density distribution 
\begin{equation}
\label{RanVar}
\Omega \Big( \nu_1,\ldots,\nu_d ,\theta_1,\ldots,\theta_d\Big)=
\dfrac{d!}{\pi^d}\prod_{i=1}^d \text{cos} \;\theta_i \;\text{sin} \;\theta_i^{2i-1}
\end{equation}
where $
\theta_i\in [0,\frac{\pi}{2} ],\,\nu_i\in [0,2 \pi ]$, 
on the set of pure quantum states parametrized as 
\begin{equation}
\label{notimp}
\ket{\psi}= \sum_{i=0}^d \Bigg( e^{i \nu_i} \;
\text{cos} \;\theta_i \prod_{\ell =i+1}^d \text{sin} \;\theta_\ell \Bigg)
\ket{i} \in \mathcal{H}_{d+1} ,
\end{equation}
with the convention $\theta_0=\nu_0=0$ \cite{bengtsson_zyczkowski_2006}. Notice that states in \cref{notimp} are normalized and 
\begin{align*}
\int_0^{2 \pi}\cdots \int_0^{2 \pi} \int_0^{\frac{\pi}{2}} \cdots \int_0^{\frac{\pi}{2}}&
\Omega \Big( \nu_1,\ldots,\nu_d ,\theta_1,\ldots,\theta_d\Big)
\text{d}\nu_1\cdot \\
&\cdot\text{d}\nu_2 \cdots \text{d}\nu_d
\text{d}\theta_1 \cdots \text{d}\theta_d
=1
\end{align*}
hence $\Omega$ it is, indeed, a probability distribution. Notice that all random variables $\theta_i,\nu_i$ in \cref{RanVar} are chosen independently according to the uniform distribution $\Theta (\nu_i )=\frac{1}{2 \pi} \nu_i$ on $[0,2 \pi]$ for $\nu_i$ variables, and according to 
\begin{equation}
\label{IndChoss}
\Omega_i \Big( \theta_i \Big)=
2\, i \, \text{cos} \;\theta_i \;\text{sin} \;\theta_i^{2i-1}
\end{equation}
distribution on $[0,\frac{\pi}{2} ]$. 

Since $\mathcal{H}_2^{\otimes n} \cong \mathcal{H}_{2^n}$, the Fubini–Study measure is also well defined on the set of pure quantum states of $n$-qubits $\mathcal{H}_2^{\otimes n} $. We shall notice some properties of such a distribution, while expressing the state in the form $\ket{\psi}=\ket{0}\ket{\psi_0} +\ket{1}\ket{\psi_1}$. For simplicity, we introduce the index sets $\mathcal{J}=\{0,1\}^n$, $\mathcal{J}_0=\mathcal{J} \setminus \{0,\ldots,0\}$, and for any index $I=(i_1,\ldots,i_n )\in \mathcal{J}$, we denote by $|I|$ the decimal representation of the binomial string $I$, i.e. $|I|=\sum_{j=1}^{n} 2^{j-1} i_j$. In order to determine the Fubini–Study measure on $\mathcal{H}_2^{\otimes n+1}$, we shall use the following isomorphism between $\mathcal{H}_2^{\otimes n+1} \cong \mathcal{H}_{2^{n+1}}$ defined on the basis vectors as $\ket{i I}\mapsto \ket{\bar{i} 2^n +|I|}$ for $I\in \mathcal{J}$, and where $\bar{i}=i+1 (\text{mod }2)$. In such a way, the Fubini–Study measure on $\mathcal{H}_2^{\otimes n+1}$ is given by the probability density distribution independent for each random variable and written jointly as
\begin{align}
\label{RanVar2}
\Omega \Big( \nu_I,\theta_I, \nu_{I'}',\theta_{I'}' : I\in \mathcal{J},I'\in \mathcal{J_0} \Big)
=
\dfrac{(2^{n+1}-1)!}{\pi^{2^{n+1}-1}}\cdot \\
\nonumber
\prod_{I\in \mathcal{J}} \text{cos} \;\theta_I \;\text{sin} \;\theta_I^{2^{n+1}+2|I|-1}
\,
\prod_{I'\in \mathcal{J_{0}}} \text{cos} \;\theta_{I'}' \;\text{sin} \;(\theta_{I'}')^{2|I'|-1}
\end{align}
where $\theta_I,\theta_I'\in [0,\frac{\pi}{2} ]$, and $
\nu_I,\nu_I'\in [0,2 \pi ],$ on the set of pure quantum states parametrized as 
\begin{equation}
\label{EqStar}
\ket{\psi}=
\ket{0}\ket{\psi_0} 
+
\big( \prod_{I\in \mathcal{J}}^{} \;\text{sin} \,\theta_I \; \big)
\ket{1}\ket{\psi_1'}
\end{equation}
where
\begin{align}
\label{notimp2}
\ket{\psi_0}= 
&\sum_{I\in \mathcal{J}} d_I \ket{I}, 
&d_I=  e^{i \nu_I} \;
\text{cos} \;\theta_I \prod_{\substack{J\in \mathcal{J}: \\ |J|>|I|}} \text{sin} \,\theta_J ,
\\
\nonumber
\ket{\psi_1'}= 
&\sum_{I\in \mathcal{J}} c_I \ket{I}, 
&c_I=  e^{i \nu_I'} \;
\text{cos} \;\theta_I' \prod_{\substack{J\in \mathcal{J}: \\ |J|>|I|}} \text{sin} \,\theta_J' ,
\end{align}
with the convention $\theta_{0\cdots 0}'=\nu_{0\cdots 0}'=0$. Notice that the state coefficients of $\ket{\psi_0}$ depends only on $\theta_I,\nu_I$ variables, while state coefficients of $\ket{\psi_1'}$ depends only on $\theta_I',\nu_I'$ variables. Furthermore, state coefficient $d_I$ of $\ket{\psi_0}$ depends only on the variables $\nu_I$ and $\theta_J$ for $J\in \mathcal{J}$ such that $|J|\geq |I|$. In particular, $d_{0\cdots 0} $ is the only coefficient in $\ket{\psi_0}$ depending on $\theta_{0\cdots 0}$, while $d_{0\cdots 0},d_{0\cdots 01} $ are the only coefficients in $\ket{\psi_0}$ depending on $\theta_{0\cdots 01}$, etc.

By a \textit{generic} pure quantum state $\ket{\psi}\in \mathcal{H}_2^{\otimes n}$ we understand a random state chosen with respect to the Fubini–Study measure on the set of quantum states \cite{bengtsson_zyczkowski_2006}. In that sense, we say that some property occurs with probability zero for generic state iff it occurs only on the set of measure zero among all pure quantum states with respect to the Fubini–Study measure. Notice that such a notion of generic states \cite{HLW06,GenericAandL,PhysRevA.93.062112} is not the only one which appear in literature. Some authors refers to generic states as those whose stabilizers meet certain symmetry conditions \cite{Gour_2017,PhysRevLett.118.040503,PhysRevA.105.032458}. It is worth mentioning that those two notions of generic states among pure quantum states agrees up to the measure zero subset on the set of pure quantum states with respect to the Fubini–Study measure.

\section{Proof of \cref{easyBusyMainText}}
\label{Section 5}
In this section, we shall prove that a generic pure quantum state will have exactly four distinct roots for each subsystem, for measures defined in \cref{(2k+1),(2k)}. We begin with the following lemma.

\begin{lemma}[Even case]
\label{lemma1odd}
Consider a generic pure quantum state $\ket{\psi}\in \mathcal{H}^{\otimes n+1}_2$ of even number of qubits written in the form $\ket{\psi}= \ket{0}\ket{\psi_0}+\ket{1}\ket{\psi_1}$  where $\ket{\psi_0},\ket{\psi_1}\in \mathcal{H}^{\otimes n}_2$, and a SLIP measure $ E^{(n)}$ presented in \cref{(2k+1)}. The conditional probability that the polynomial equation $ E^{(n)} \big(\ket{\psi_0}+z\ket{\psi_1}\big)=0$ has no multiple root at $z=0$ under the condition that it has at least a single root at $z=0$ is zero.
\end{lemma}

\begin{proof}[Proof of \cref{lemma1odd}]
Consider an $(n+1)$-qubit state $\ket{\psi}\in \mathcal{H}^{\otimes n+1}_2$ written in the form $\ket{\psi}= \ket{0}\ket{\psi_0}+\ket{1}\ket{\psi_1}$, and suppose that $n=2k+1$ is an odd number. We shall evaluate the entanglement measure presented in \cref{(2k+1)} on the following state $\ket{\psi_0}+z\ket{\psi_1}$, i.e.
\begin{align*}
E^{(2k+1)}& \big(\ket{\psi_0}+z\ket{\psi_1}\big) =\\
&=\Bigg\vert 
\sum_{j=\Id,x,y,z} \eta_j \;
\Big(
\bra{\psi_0}\sigma_j \otimes \sigma_y \otimes\cdots\otimes \sigma_y
\ket{\bar{\psi_0}} +\\
&+2 z  \bra{\psi_1}\sigma_j \otimes \sigma_y \otimes\cdots\otimes \sigma_y
\ket{\bar{\psi_0}} 
\\
&+z^2 \bra{\psi_1}\sigma_j \otimes \underbrace{\sigma_y \otimes\cdots\otimes \sigma_y}_{2k} 
\ket{\bar{\psi_1}} 
\Big)^2
\Bigg\vert .
\nonumber
\end{align*}
Furthermore, we expand the equation above as an absolute value of a degree four polynomial in $z$ variable, namely
\begin{equation*}
\label{polypoly}
E^{(2k+1)} \big(\ket{\psi_0}+z\ket{\psi_1}\big) =
\big\vert 
C_0 +C_1 z+ C_2 z^2 +C_3 z^3+ C_4 z^4
\big\vert 
\end{equation*} 
where the coefficients $C_i$ are of the following form:
\begin{align*}
C_0
=& 
\sum_{j=\Id,x,y,z} \eta_j \;
\bra{\psi_0}\sigma_j \otimes \sigma_y \otimes\cdots\otimes \sigma_y
\ket{\bar{\psi_0}} 
\cdot
\\&\cdot
\bra{\psi_0}\sigma_j \otimes \sigma_y \otimes\cdots\otimes \sigma_y
\ket{\bar{\psi_0}} 
=
E^{(2k+1)} \big(\ket{\psi_0}\big),
\\
C_1
=& 2
\sum_{j=\Id,x,y,z} \eta_j \;
\bra{\psi_1}\sigma_j \otimes \sigma_y \otimes\cdots\otimes \sigma_y
\ket{\bar{\psi_0}} 
\cdot
\\&\cdot
\bra{\psi_0}\sigma_j \otimes \sigma_y \otimes\cdots\otimes \sigma_y
\ket{\bar{\psi_0}} ,
\\
&\ldots 
\end{align*}
Notice that the equation $ E^{(n)} \big(\ket{\psi_0}+z\ket{\psi_1}\big)=0$ has root at $z=0$ iff the coefficient $C_0 =0$, and has a multiple root at $z=0$ iff coefficients $C_0 ,C_1=0$. Therefore we shall prove that the conditional probability
\begin{equation}
\label{conditional1}
\mathbb{P} \Big(C_1=0 \,|\, C_0=0 \Big) =0 
\end{equation}
that $C_1=0$ under the condition $C_0=0$ is zero. 

We can rewrite the coefficients $C_0 ,C_1$ in the following way
\begin{align*}
C_0 & = \braket{\psi_0 | A(\psi_0)}\,, \nonumber \\
C_1 & = \braket{\psi_1 | A(\psi_0)}\,,
\end{align*}
where $\ket{A(\psi_0)}$ is the following vector
\begin{equation*}
\ket{A(\psi_0)} = A(\psi_0) \ket{\overline{\psi}_0}\,,
\end{equation*}
defined by operator 
\begin{align*}
A(\psi_0) = 
\sum_{j=\Id,x,y,z} &\left(\eta_j \braket{\psi_0|\si_j \otimes \si_y \cdots \otimes \si_y |\overline{\psi}_0}\right) \cdot\\
&\cdot\si_j \otimes \si_y \cdots \otimes \si_y\,.
\end{align*}
Notice that $\ket{A(\psi_0)} $ is a function of state coefficients of $\ket{\psi_0}$ only. State coefficients of $\ket{\psi_1}$ are not related to the state coefficients of $\ket{\psi_0}$ in anyway by a norm, see \cref{EqStar,notimp2}, hence the space of solutions for the equation
\begin{equation}
C_1 = \braket{\psi_1 | A(\psi_0)} = 0
\end{equation}
is of zero measure iff $\ket{A(\psi_0)} \not\equiv 0$. Therefore, condition \cref{conditional1} can be rewritten as the conditional probability
\begin{equation}
\label{conditional2}
\mathbb{P} \Big(\ket{A(\psi_0)} \equiv 0 \,|\, 
\braket{\psi_0 | A(\psi_0)} =0 \Big) =0 
\end{equation}
that $\ket{A(\psi_0)} \equiv 0$ under the condition $\braket{\psi_0 | A(\psi_0)} =0$ is zero. 

For simplicity, we shall introduce the following notation:
\begin{equation*}
\be_k = \braket{\psi_0|\si_k \otimes \si_y \cdots \otimes \si_y |\overline{\psi}_0}\,,
\end{equation*}
for $i=0,2,3$. It is easy to see that $\braket{\psi_0 | A(\psi_0)} = -\be^2_0+\be^2_2+\be^2_3$, hence the condition $\braket{\psi_1 | A(\psi_0)} =0$ can be rewritten as 
\begin{equation}
\label{cond1}
\sum_k \eta_k \be^2_k = -\be^2_0+\be^2_2+\be^2_3 = 0\,,
\end{equation}
Furthermore, the operator $A(\psi_0)$ takes the following form
\begin{equation*}
A(\psi_0) = \left(\sum_{\mu} \eta_{\mu} \be_{\mu} \si_{\mu}\right) \otimes \si_y \otimes \cdots \otimes \si_y\,.
\end{equation*}
Observe that condition $\ket{A(\psi_0)}=A(\psi_0) \ket{\overline{\psi}_0}\equiv 0$ is equivalent to requiring that $\ket{\overline{\psi}_0}$ belongs to the null space $\ma{N}(A(\psi_0))$ of the matrix $A(\psi_0)$. Therefore, we have an immediate implication: if the null space of the operator $A(\psi_0)$ is trivial, i.e. $\ma{N}(A(\psi_0))=\{0\}$ then $\ket{A(\psi_0)}\not\equiv 0$. Since the null space of $\si_y$ is empty, so is the null space of any number of tensor products of $\si_y$. Hence, we shall find a null space $\ma{N}\left(\sum_{\mu} \eta_{\mu} \be_{\mu} \si_{\mu}\right)$. One simple way to do this is to calculate the eigenvectors of $\sum_{\mu} \eta_{\mu} \be_{\mu} \si_{\mu}$ and look at the space spanned by the eigenvectors whose eigenvalues are 0. It is straightforward to show that the eigenvalues and eigenvectors of a general combination of Pauli matrices $\sum_{\mu} a_{\mu} \si_{\mu}$ are given by
\begin{equation*}
\label{eigenSystem}
v_{\pm} = \left\{
\begin{pmatrix}
\frac{a_3 \pm \sqrt{a^2_1 + a^2_2 + a^2_3}}{a_1 + i a_2} \\
1
\end{pmatrix}
\right\}\,, \\
\lambda_{\pm} = a_0 \pm \sqrt{a^2_1 + a^2_2 + a^2_3}\,.
\end{equation*}
Hence, for $a_{\mu} = \eta_{\mu} \be_{\mu}$ under the condition $C_0=0$ (in form \cref{cond1}), we have the eigensystem
\begin{equation*}
\label{eigenSystemD}
v_{\pm} = \left\{
\begin{pmatrix}
-i \frac{\be_3 \pm \sqrt{\be^2_0}}{\be_2} \\
1
\end{pmatrix}
\right\}\,, \\
\lambda_{\pm} = \be_0 \pm \sqrt{\be^2_0}\,.
\end{equation*}
Therefore if $\be_0 \neq 0$ then $\ket{A(\psi_0)} \not\equiv 0$. Recall that the condition $C_0 = 0$ is equivalent to $\be^2_0 = \be^2_2 + \be^2_3$, hence the following conditional probability
\begin{equation}
\label{conditional3}
\mathbb{P} \Big(\be_0 = 0 
\,|\, 
\be^2_0 = \be^2_2 + \be^2_3 \Big) =0 
\end{equation}
that $\be_0 = 0$ under the condition $\be^2_0 = \be^2_2 + \be^2_3$ is zero implies \cref{conditional2}. In the remaining part of the proof, we shall show that \cref{conditional3} holds true.

We begin by expanding the $n$-qubit state $\ket{\psi_0}$ in the computational basis 
\begin{equation*}
\ket{\psi_0} = \sum_{j_1,\ldots,j_{n}=0}^n d_{j_1  \, \cdots \,  j_{n}} \ket{j_1 ,\, \ldots \, , j_{n}}\,.
\end{equation*}
Notice that $\be_k$ takes the following form in the computational basis
\begin{align*}
\be_k = \sum_{\substack{j_1,\ldots,j_{n} \\i_1,\ldots,i_{n} }} &d^{}_{j_1 \, \cdots \,  j_{n}} d_{i_1  \, \cdots \,  i_{n}} \braket{j_1|\si_k|i_1}\cdot\\
&\cdot \braket{j_2, \, \ldots \, , j_{n}|\si_y \otimes\cdots \otimes \si_y |i_2 ,\,  \ldots \, , i_n}\,.
\end{align*}
Notice that because of the cancelation $\be_2=0$, we can further simplify the expression for $\be_k$ and single out one state coefficient, namely $d_{0 \,\cdots \,0}$. We have
\begin{align}
\nonumber
\be_0 & = 2 (-1)^k\, (  D_0 \,+\, D_3),
\\
\label{besimple}
\be_2 & = 0 ,
\\
\nonumber
\be_3 & = 2(-1)^k \, ( D_0 \,-\, D_3) ,
\end{align}
where
\begin{align}
\nonumber
D_0 =&
\sum_{I \in \mathcal{I} } (-1)^{|I|} d_{0\,I}d_{0 \bar{I}}\,,
&
\\
\label{DDADD}
D_1 =&
\sum_{I \in \mathcal{I}} (-1)^{|I|} d_{1\,I}d_{1 \bar{I}}\,,
\end{align}
where $\mathcal{I}=\{0,1\}^{n-1}$ is the index set introduced for consistency of the notation, and for any bit-string $I=(i_1,\ldots,i_{n-1})$, we define its complement $\bar{I}=$ as $\bar{I}:= (\bar{i_1},\ldots,\bar{i_{n-1}})$ where $\bar{i_j} = i_j+1 \,(\text{mod}\, 2)$, and $|I|=\sum_{j=1}^{n-1} i_j$. 
Expanding $\be_k$ according to \cref{besimple}, the condition $\be^2_0 = \be^2_2 + \be^2_3$ becomes the following
\begin{equation}
\big( D_0\,+ \,D_1 \big)^2 
\,=\,
\big( D_0\,- \,D_1 \big)^2
\end{equation} 
and hence has exactly two solutions, either $D_0 =0$ or $D_1=0$. Notice that under any of those two solutions the condition $\beta_0=0$, i.e $D_0+D_1=0$ becomes $D_0=D_1=0$. Hence  \cref{conditional3} can be rewritten as the following conditional probability
\begin{align}
\label{conditional4a}
\mathbb{P} \Big( D_0=0\,\wedge\, D_1=0 
\,\Big|\, 
D_0=0\,\vee \, D_1=0 
\Big) =0 .
\end{align}
Notice that
\begin{align}
\label{conditional5}
\mathbb{P} &\Big( D_0=0\,\wedge\, D_1=0 
\,\Big|\, 
D_0=0\,\vee \, D_1=0 
\Big) \leq \\
\nonumber
&
\mathbb{P} \Big( D_0=0
\,\Big|\, 
 D_1=0 
\Big)
+
\mathbb{P} \Big( D_1=0
\,\Big|\, 
 D_0=0 
\Big)
,
\end{align}
hence we shall show that both terms on the right-hand side of the equation above vanish. 

Note that equation $D_0=0$ imposes conditions on the state coefficients $d_{0 I}$, $I\in \mathcal{I}$ only, while equation $D_1=0$ on the state coefficients $d_{1 I}$, $I\in \mathcal{I}$, see \cref{DDADD}. Since equations $D_1=0$ and $ D_0=0$ impose conditions on different state coefficients, intuitively condition $ D_0=0$ should not enforce $ D_1=0$ being satisfied (and vice-versa). This can be rigorously showed by the properties of Fubini–Study measure. Recall that state coefficients $d_{1 I} $ in $\ket{\psi_0}$ depend on $\theta_{1 I}$, and $ \nu_{1 I}$ where $I\in \{0,1\}^{n-1}$ only, see \cref{EqStar,notimp2}. Therefore, the condition $ D_1=0$ can be rewritten in $\theta_{1 I}$ and $ \nu_{1 I}$ variables as
\begin{equation}
\label{evFryA}
f(\theta_{1I},\nu_{1I} : I\in \mathcal{I} ) =0,
\end{equation}
where $f$ is an elementary function in $\theta_{1I},\nu_{1I} : I\in \mathcal{I} $ variables. The exact form of $f$ can be traced back from \cref{DDADD} and \cref{notimp2}. In particular $f$, as an elementary function, is continuous on its domain. Similarly, the condition $ D_0=0$ can be rewritten in $\theta_{0 I}$ and $ \nu_{0 I}$ variables as
\begin{equation}
\label{evFryB}
g(\theta_{0I},\nu_{0I} : I\in \mathcal{I} ) =0,
\end{equation}
where $g$ is an elementary, hence continuous, function. Therefore, $\mathbb{P} ( D_1=0 \,|\, D_0=0 )$ can be rewritten as the following conditional probability
\begin{align}
\label{conditional6a}
\mathbb{P} \Big( \,
f(\theta_{1I},\nu_{1I} : I\in \mathcal{I} ) =0
\,\Big|\, 
g(\theta_{0I},\nu_{0I} : I\in \mathcal{I} ) =0
\,\Big) =0 
\end{align}
where $f,g$ are elementary (in particular continuous) functions on its domains defined on different variables. Recall that the random variables $\theta_{I},\nu_I : I\in \{0,1\}^n$ were chosen independently with continuous probability density distributions described in \cref{RanVar2}, i.e $\nu_I$ are chosen independently according to the uniform distribution $\Theta (\nu_I )=\frac{1}{2 \pi} \nu_I$ on $[0,2 \pi]$, and $\theta_I$ are chosen according to 
\begin{equation}
\label{IndChoss22a}
\Omega_I \Big( \theta_I\Big)=
\big(2^{n+1}\,+2\,|I|\big) \, \text{cos} \;\theta_I \;\text{sin} \;\theta_I^{2^{n+1}\,+2\,|I|\,\,-1}
\end{equation}
distribution on $[0,\frac{\pi}{2} ]$. Therefore, a random quantum state $\ket{\psi_0}$ which satisfies condition $g(\theta_{0I},\nu_{0I} : I\in \mathcal{I} ) =0$ is given by randomly chosen values of $\theta_I,\nu_I : I\in \{0,1\}^n$ with respect to the probability distributions in \cref{IndChoss22a} for which $g(\theta_{0I},\nu_{0I} : I\in \mathcal{I} ) =0$. It does not impose, however, any further constrains on $\nu_{1 I} \,\theta_{1I}$ coefficients. Hence the value of $f(\theta_{1I},\nu_{1I} : I\in \mathcal{I} )$ under condition $g(\theta_{0I},\nu_{0I} : I\in \mathcal{I} ) =0$ does not vanish except of measure zero subspace for such induced probability distribution, hence $\mathbb{P} ( D_1=0 |  D_0=0 )=0$. An analogous argument shows that $\mathbb{P} ( D_0=0 |  D_1=0 )=0$, and hence that \cref{conditional4a} holds true, which finishes the proof of \cref{lemma1odd}.
\end{proof}

Below, we prove the result analogous to \cref{lemma1odd}, for any odd number of qubits. 

\begin{lemma}[Odd case]
\label{lemma1even}
Consider a generic pure quantum state $\ket{\psi}\in \mathcal{H}^{\otimes n+1}_2$ of odd number of qubits written in the form $\ket{\psi}= \ket{0}\ket{\psi_0}+\ket{1}\ket{\psi_1}$  where $\ket{\psi_0},\ket{\psi_1}\in \mathcal{H}^{\otimes n}_2$, and a SLIP measure $ E^{(n)}$ presented in \cref{(2k)}. The conditional probability that the polynomial equation $ E^{(n)} \big(\ket{\psi_0}+z\ket{\psi_1}\big)=0$ has no multiple root at $z=0$ under the condition that it has at least single root at $z=0$ is zero.
\end{lemma}

\begin{proof}
Consider an $n$-qubit state $\ket{\psi}\in \mathcal{H}^{\otimes n}_2$ written in the form $\ket{\psi}= \ket{0}\ket{\psi_0}+\ket{1}\ket{\psi_1}$, and suppose that $n=2k$ is an even number. We shall evaluate the value of a measure $E^{(2k)}$ presented in \cref{(2k)} on the following state $\ket{\psi_0}+z\ket{\psi_1}$. Thus, we have
\begin{align}
\label{sandyBeach2}
E^{(2k)} &\big(\ket{\psi_0}+z\ket{\psi_1}\big) =\\
\nonumber &
\Bigg\vert 
\sum_{j,i=\Id,x,y,z} \eta_j \nu_i \;
\Big(
\bra{\psi_0}\sigma_j \otimes \sigma_i \otimes \sigma_y \otimes\cdots\otimes \sigma_y
\ket{\bar{\psi_0}} \\
\nonumber
&+2 z  \bra{\psi_1}\sigma_j \otimes \sigma_i \otimes \sigma_y \otimes\cdots\otimes \sigma_y
\ket{\bar{\psi_0}} 
\\
\nonumber
&+z^2 \bra{\psi_1}\sigma_j \otimes \sigma_i \otimes \underbrace{\sigma_y \otimes\cdots\otimes \sigma_y}_{2k} 
\ket{\bar{\psi_1}} 
\Big)^2
\Bigg\vert 
\nonumber
\end{align}
Observe that by expanding it further, it is an absolute value of a degree four polynomial in the $z$ variable:
\begin{equation*}
E^{(2k)} \big(\ket{\psi_0}+z\ket{\psi_1}\big) =
\big\vert 
C_0 +C_1 z+ C_2 z^2 +C_3 z^3+ C_4 z^4
\big\vert 
\end{equation*} 
where the coefficients are of the following form:
\begin{align*}
C_0
&= 
\sum_{j,i=\Id,x,y,z} \eta_j \nu_i\;
\bra{\psi_0}\sigma_j \otimes \sigma_i \otimes \sigma_y \otimes\cdots\otimes \sigma_y
\ket{\bar{\psi_0}} 
\cdot\\
&\cdot
\bra{\psi_0}\sigma_j \otimes \sigma_i \otimes \sigma_y \otimes\cdots\otimes \sigma_y
\ket{\bar{\psi_0}} 
=
E^{(2k)} \big(\ket{\psi_0}\big),
\\
C_1
&= 2
\sum_{j,i=\Id,x,y,z} \eta_j \nu_i\;
\bra{\psi_1}\sigma_j \otimes \sigma_i \otimes \sigma_y \otimes\cdots\otimes \sigma_y
\ket{\bar{\psi_0}} 
\cdot\\
&\cdot
\bra{\psi_0}\sigma_j \otimes \sigma_i \otimes \sigma_y \otimes\cdots\otimes \sigma_y
\ket{\bar{\psi_0}} ,
\\
&\ldots 
\end{align*}
with $(\eta_\Id,\eta_x,\eta_y,\eta_z )=(\nu_\Id,\nu_x,\nu_y,\nu_z ) =( -1,0,1,1)$. Notice that the equation $ E^{(n)} \big(\ket{\psi_0}+z\ket{\psi_1}\big)=0$ has a root at $z=0$ iff the coefficient $C_0 =0$, and has a multiple root at $z=0$ iff coefficients $C_0 ,C_1=0$. Therefore we shall prove that the conditional probability
\begin{equation}
\label{conditional1bis}
\mathbb{P} \Big(C_1=0 \,|\, C_0=0 \Big) =0 
\end{equation}
that $C_1=0$ under the condition $C_0=0$ is zero. 

We can rewrite the coefficients $C_0 ,C_1$ in the following way
\begin{align*}
C_0 & = \braket{\psi_0 | A(\psi_0)}\,, \nonumber \\
C_1 & = \braket{\psi_1 | A(\psi_0)}\,,
\end{align*}
where $\ket{A(\psi_0)}$ is the following vector
\begin{equation}
\label{AB1}
\ket{A(\psi_0)} = A(\psi_0) \ket{\overline{\psi}_0}\,,
\end{equation}
defined by operator 
\begin{align*}
A(\psi_0) = \sum_{i,j=Id,x,y,z} &\left(\eta_j \eta_i \braket{\psi_0|\si_i \otimes \si_j \otimes \si_y \cdots \otimes \si_y |\overline{\psi}_0}\right) \\
&\cdot\si_i \otimes \si_j \otimes \si_y \cdots \otimes \si_y\,,
\end{align*}
Notice that $\ket{A(\psi_0)} $ is a function of state coefficients of $\ket{\psi_0}$ only. State coefficients of $\ket{\psi_1}$ are not related to the state coefficients of $\ket{\psi_0}$ in anyway by a norm, see \cref{EqStar,notimp2}, hence the space of solutions for the equation
\begin{equation}
C_1 = \braket{\psi_1 | A(\psi_0)} = 0
\end{equation}
is of zero measure iff $\ket{A(\psi_0)} \not\equiv 0$. Therefore, condition \cref{conditional1bis} can be rewritten as the conditional probability
\begin{equation}
\label{conditional2bis}
\mathbb{P} \Big(\ket{A(\psi_0)} \equiv 0 \,|\, 
\braket{\psi_0 | A(\psi_0)} =0 \Big) =0 
\end{equation}
that $\ket{A(\psi_0)} \equiv 0$ under the condition $\braket{\psi_0 | A(\psi_0)} =0$ is zero. Observe that in order to show that $\ket{A(\psi_0)}=A(\psi_0) \ket{\overline{\psi}_0}\not\equiv 0$ it is enough to show that the first coefficient of $\ket{A(\psi_0)}$ is non-vanishing, i.e. $\braket{0\ldots0|A(\psi_0)} \neq 0$. Therefore the following conditional probability
\begin{equation}
\label{conditional3bis}
\mathbb{P} \Big(\braket{0\ldots0|A(\psi_0)} = 0
\,|\, 
\braket{\psi_0 | A(\psi_0)} =0 \Big) =0 
\end{equation}
that $\braket{0\ldots0|A(\psi_0)} = 0$ under the condition $\braket{\psi_0 | A(\psi_0)} =0$ is zero implies \cref{conditional2bis}. In the remaining part of the proof, we shall show that \cref{conditional3bis} holds true.

For simplicity, we shall introduce the following notation:
\begin{equation*}
\be_{kp} = \braket{\psi_0|\si_k \otimes \si_p \otimes \si_y \cdots \otimes \si_y |\overline{\psi}_0}\,,
\end{equation*}
for $i,j=0,2,3$. It is easy to see that the condition $\braket{\psi_0 | A(\psi_0)} =0$ can be rewritten as 
\begin{align}
\label{cond1bis}
0&=\sum_{k,p} \eta_k \eta_p \be^2_{kp} = 
\\
\nonumber
&
\be^2_{00} - \be^2_{02} - \be^2_{03} - \be^2_{20} + \be^2_{22} + \be^2_{23} -\be^2_{30} +\be^2_{32}+\be^2_{33},
\end{align}
Furthermore, the operator $A(\psi_0)$ takes the following form
\begin{align}
\label{DoubleStar}
A(\psi_0) =& \left(\sum_{k,p} \eta_{k} \eta_{p} \be_{kp} \si_{k} \otimes \si_{p}\right) \otimes \si_y \otimes \cdots \otimes \si_y = \\
\nonumber
=&B \otimes \si_y \otimes \cdots \otimes \si_y\,,
\end{align}
where matrix $B$ is presented on \cref{fig:4}.

\begin{figure*}
\begin{equation*}
B =
\begin{pmatrix}
 \beta_{00}-\beta_{03}-\beta_{30}+\beta_{33} & i \left(\beta_{02}-\beta_{32}\right) & i \left(\beta_{20}-\beta_{23}\right) & -\beta_{22} \\
 -i \left(\beta_{02}-\beta_{32}\right) & \beta_{00}+\beta_{03}-\beta_{30}-\beta_{33} & \beta_{22} & i \left(\beta_{20}+\beta_{23}\right) \\
 -i \left(\beta_{20}-\beta_{23}\right) & \beta_{22} & \beta_{00}-\beta_{03}+\beta_{30}-\beta_{33} & i \left(\beta_{02}+\beta_{32}\right) \\
 -\beta_{22} & -i \left(\beta_{20}+\beta_{23}\right) & -i \left(\beta_{02}+\beta_{32}\right) & \beta_{00}+\beta_{03}+\beta_{30}+\beta_{33} \\
\end{pmatrix} \,.
\end{equation*}
\caption{The form of matrix $B$ from \cref{DoubleStar}.}
\label{fig:4}
\end{figure*}


We shall expand an $n$-qubit state $\ket{\psi_0}$ in the computational basis 
\begin{equation*}
\ket{\psi_0} = \sum_{j_1,\ldots,j_{n}=0}^n d_{j_1  \, \cdots \,  j_{n}} \ket{j_1 ,\, \ldots \, , j_{n}}\,.
\end{equation*}
Notice that $\be_{kp}$ takes the following form in the computational basis
\begin{align*}
\be_{kp} = &\sum_{\substack{j_1,\ldots,j_{n} \\i_1,\ldots,i_{n} }} d^{}_{j_1 \, \cdots \,  j_{n}} d_{i_1  \, \cdots \,  i_{n}} \braket{j_1|\si_k|i_1} \braket{j_2|\si_p|i_2} \cdot\\
&\cdot\braket{j_3, \, \ldots \, , j_{n}|\si_y \otimes\cdots \otimes \si_y |i_3 ,\,  \ldots \, , i_n}\,.
\end{align*}
We can further simplify the expression for $\be_{kp}$:
\begin{align}
\label{besimplebis}
\be_{00} & = (-1)^{k-1}(D_{00} \,+\, D_{01}\,+\,D_{10} \,+\, D_{11}),
\\
\nonumber
\be_{03} & = (-1)^{k-1}(D_{00} \,-\, D_{01}\,+\,D_{10} \,-\, D_{11}), 
\\
\nonumber
\be_{30} & = (-1)^{k-1}(D_{00} \,+\, D_{01}\,-\,D_{10} \,-\, D_{11}), 
\\
\nonumber
\be_{33} & = (-1)^{k-1}(D_{00} \,-\, D_{01}\,-\,D_{10} \,+\, D_{11}), 
\\
\nonumber
\be_{22} & = 2(-1)^{k} (B_0-B_1),
\\
\be_{02} & = \be_{32}= \be_{20}= \be_{23}=0,
\nonumber
\end{align}
where we define the following quantities
\begin{align*}
& D_{ij}= \sum_{I \in \mathcal{I} } (-1)^{|I|} d_{ij\,I}d_{ij \bar{I}}\,,
\\ \nonumber
& B_{0}= \sum_{I \in \mathcal{I} } (-1)^{|I|} d_{00\,I}d_{11 \bar{I}}\,,
\\ \nonumber
& B_{1}= \sum_{I \in \mathcal{I} } (-1)^{|I|} d_{01\,I}d_{10 \bar{I}}\,,
\nonumber
\end{align*}
and $\mathcal{I}=\{0,1\}^{n-2}$ is the index set introduced for consistency of the notation, and for $I=(i_1,\ldots,i_{n-2})$, we define $\bar{I}= (\bar{i_1},\ldots,\bar{i_{n-2}})$ where $\bar{i_j} = i_j+1 \,(\text{mod}\, 2)$, and $|I|=\sum_{j=1}^{n-2} i_j$. With above notation at hand, \cref{cond1bis} becomes the following 
\begin{equation}
\label{purple}
8\,(D_{01}D_{10}\,+\, D_{11}D_{00})\,-\,2\,(B_0\, -\, B_1)
\, =\, 0.
\end{equation}
We shall single out $d_{0\cdots 0}$ state coefficient from the equation above. Notice that the state coefficient $d_{0\ldots 0}$ appears in $D_{00}$ and $B_0$ terms only. Denote by
\begin{align}
\nonumber
\bar{D}_{00} =& D_{00} \,-\, 2d_{0\cdots 0} d_{001\cdots 1},
\\
\nonumber
\hat{D}_{00} =& \bar{D}_{00}\,+\,2 d_{0\cdots 01} d_{001\cdots 10
},
\\
\nonumber
\bar{B}_{0} =& B_{0} \,-\,2 d_{0\cdots 0} d_{1\cdots 1},
\\
\hat{B}_{0} =& \bar{B}_{0} \,+\, 2d_{0\cdots 01} d_{1\cdots 10
}.
\label{AB2}
\end{align}
Observe that $d_{0\cdots 0}$ appears in $D_{00}$ and $B_0$, but does not appear in any of the following terms $\bar{D}_{00},D_{01},D_{10},D_{11}, \bar{B}_{0},B_{1}$. Hence \cref{purple} might be solved with respect to $d_{0\cdots 0}$ variable, and become
\begin{equation}
d_{0\ldots 0} =
\dfrac{4(D_{01} D_{10} +D_{11} \bar{D}_{00})-(\bar{B}_{0}-B_1)}{ 2 (d_{1\cdots 1}-4 D_{11} d_{001\cdots 1})}.
\label{regardless}
\end{equation}
In such a way, we expressed the condition $\braket{\psi_0 | A(\psi_0)} =0 $
which appears in \cref{conditional3bis} as an equation satisfied by the $d_{0\cdots 0}$ state coefficient. As a next step, we shall investigate thr following equation
\begin{equation}
\label{AB3}
\braket{0\ldots0|A(\psi_0)} = 0.
\end{equation}
From (\ref{DoubleStar}), (\ref{AB1}), and (\ref{besimplebis}), we have
\begin{equation*}
\braket{ A(\psi_0)|0\ldots0}=
4 D_{11} d_{001\cdots 1 } +2 (B_{0}- B_1) d_{1 \cdots 1} .
\end{equation*}
Notice that in the equation above the $d_{0\cdots 0}$ coefficient appears only in $B_0$ term. By expressing $B_0$ according to (\ref{AB2}), and substituting $d_{0\cdots 0}$ according to (\ref{regardless}), \cref{AB3} becomes
\begin{align}
\label{AB5}
0\,&=\,4 d_{001\cdots 1} D_{11} +8 d_{1\cdots 1}
\dfrac{}
{d_{1\cdots 1} -4 d_{001\cdots 1}D_{11}} \cdot
\\
\nonumber
&\cdot\Big({d_{1\cdots 1} D_{01} D_{10} +D_{11} (d_{001\cdots 1} (B_1 -\bar{B}_0 ) +d_{1\cdots 1} \bar{D}_{00})}
\Big)
.
\end{align}
Notice that in the equation above, the coefficient $d_{0\cdots 01}$ appears only in $\bar{D}_{00}$ and $\bar{B}_0$ terms. With the notation (\ref{AB2}), we can single out the $d_{0\cdots 01}$ coefficient and solve \cref{AB5} with respect to $d_{0\cdots 01}$ variable:
\begin{align}
\nonumber
d_{0\cdots 01}&=
\dfrac{2 d_{1\cdots 1}^2 D_{01} D_{10} +D_{11} }
{4d_{1\cdots 1}D_{11} (d_{001\cdots 10}d_{1\cdots 1}-d_{001\cdots 1}d_{1\cdots 10})}
\cdot\\
\label{DarkNightII}
&
\cdot
\Big((d_{001\cdots1} d_{1\cdots 1}(2 B_1+1)
\\
\nonumber
&-
2d_{001\cdots1} 
(d_{1\cdots 1} \hat{B}_0+ 2 d_{00 1\cdots 1}D_{11} ) +2 d_{1\cdots 1}^2 
\hat{D}_{00}\Big)
\end{align}
Notice that with \cref{regardless,DarkNightII} in hand, \cref{conditional3} can be rewritten as the following conditional probability
\begin{align}
\label{conditional4}
\mathbb{P} \Big( 
\text{ \cref{DarkNightII} holds }
\,\Big|\, 
\text{ \cref{regardless} holds }
\Big) =0 
\end{align}
that the probability that \cref{DarkNightII} holds under the condition that \cref{regardless} holds is zero. 

Note that \cref{regardless} expresses state coefficient $d_{0\cdots 0}$ as a function of other state coefficients, and it defines a subspace in the state space in which it holds. Intuitively, there should be no restriction on other state coefficients, as \cref{DarkNightII}, in such a subspace. Hence \cref{conditional4} should be satisfied. This intuition can be rigorously shown by properties of the Fubini–Study measure. In the remaining part of a proof, we show that \cref{conditional4} holds, indeed, true.
Notice that $d_{0 \cdots 0}$ is on the left side of \cref{regardless} and does not appear on the right side of \cref{regardless}. We shall recall now some properties of Fubini–Study distribution on the set $\mathcal{H}_2^{\otimes n+1}$ presented in \cref{EqStar,notimp2}. Recall that, $d_{0\cdots 0} $ is the only coefficient in $\ket{\psi_0}$ depending on $\theta_{0\cdots 0}$. Therefore, by multiplying \cref{DarkNightII} by $\prod_{I\in \mathcal{J}_0} \frac{1}{\text{sin}\, \theta_I}$, where we used notation $\mathcal{J}_0=\{0,1\}^{n} \setminus \{0,\ldots,0\}$, \cref{DarkNightII} takes the following form
\begin{equation}
\label{evFry}
e^{i\,\nu_{0\cdots 0}} \,\text{cos}\, \theta_{0\cdots 0} =
f(\theta_I,\nu_I : I\in \mathcal{J}_0),
\end{equation}
where $f$ is an elementary function in $\theta_I,\nu_I : I\in \mathcal{J}_0$ variables, where exact form of $f$ can be traced back from \cref{DarkNightII} by substituting successively \cref{DDADD} and \cref{notimp2}. In particular $f$, as an elementary function, is continuous on its domain. 
Similarly, the coefficient $d_{0 \cdots 01}$ is on the left side of \cref{regardless} and does not appear on the right side of \cref{regardless}, moreover there is no $d_{0 \cdots 0}$ coefficient in both sides of \cref{regardless}. As we noticed in \cref{notimp2}, while $d_{0\cdots 0},d_{0\cdots 01} $ are the only coefficients in $\ket{\psi_0}$ depending on $\theta_{0\cdots 01}$, therefore, by multiplying \cref{regardless} by $\prod_{I\in \mathcal{J}_1} \frac{1}{\text{sin}\, \theta_I}$, where we used notation $\mathcal{J}_1=\mathcal{J}_0 \setminus \{0,\ldots,0,1\}$, \cref{regardless} takes the following form
\begin{equation}
\label{evFry2}
e^{i\,\nu_{0\cdots 01}} \,\text{cos}\, \theta_{0\cdots 01} =
g(\theta_I,\nu_I : I\in \mathcal{J}_1),
\end{equation}
where $g$ is an elementary function in $\theta_I,\nu_I : I\in \mathcal{J}_1$ variables, where exact form of $g$ can be traced back from \cref{regardless} by substituting successively \cref{DDADD} and \cref{notimp2}. In particular $g$, as an elementary function, is continuous on its domain. 
In summary, \cref{conditional4} can be rewritten as the following conditional probability
\begin{align}
\label{conditional6}
\mathbb{P} \Big( \,
e^{i\,\nu_{0\cdots 01}}& \,\text{cos}\, \theta_{0\cdots 01} =
g(\theta_I,\nu_I : I\in \mathcal{J}_1)
\,\Big|\, 
\\
\nonumber
&
\,\Big|\, 
e^{i\,\nu_{0\cdots 0}} \,\text{cos}\, \theta_{0\cdots 0} =
f(\theta_I,\nu_I : I\in \mathcal{J}_0)
\,\Big) =0 
\end{align}
that the probability that $e^{i\,\nu_{0\cdots 01}} \,\text{cos}\, \theta_{0\cdots 01} =g(\theta_I,\nu_I : I\in \mathcal{J}_1)$ under the condition that $e^{i\,\nu_{0\cdots 0}} \,\text{cos}\, \theta_{0\cdots 0} =f(\theta_I,\nu_I : I\in \mathcal{J}_0)$ is zero, where $f,g$ are elementary (in particular continuous)  functions on its domains in $\theta_I,\nu_I : I\in \mathcal{J}_0$ and $\theta_I,\nu_I : I\in \mathcal{J}_1$ variables respectively. Recall that the random variables $\theta_I,\nu_I : I\in \mathcal{J} $ were chosen independently with continuous probability density distributions described in \cref{RanVar2}, i.e $\nu_I$ are chosen independently according to the uniform distribution $\Theta (\nu_I )=\frac{1}{2 \pi} \nu_I$ on $[0,2 \pi]$, and $\theta_I$ are chosen according to 
\begin{equation}
\label{IndChoss22}
\Omega_I \Big( \theta_I\Big)=
\big(2^{n+1}\,+2\,|I|\big) \, \text{cos} \;\theta_I \;\text{sin} \;\theta_I^{2^{n+1}\,+2\,|I|\,\,-1}
\end{equation}
distribution on $[0,\frac{\pi}{2} ]$. Therefore, a random quantum state $\ket{\psi_0}$ which satisfies condition $e^{i\,\nu_{0\cdots 0}} \,\text{cos}\, \theta_{0\cdots 0} =f(\theta_I,\nu_I : I\in \mathcal{J}_0)$ is given by randomly chosen values of $\theta_I,\nu_I : I\in \mathcal{I}$ with respect to the probability distributions in \cref{IndChoss22} for which the norm of $f(\theta_I,\nu_I : I\in \mathcal{J}_0)$ is smaller then one, and coefficients $\nu_{0\cdots 0}, \theta_{0\cdots 0}$ uniquely determined by the value of function $f$. It does not impose, however, any further constrains on $\nu_{0\cdots 01} \,\theta_{0\cdots 01}$ coefficient in terms od $\theta_I,\nu_I : I\in \mathcal{J}_1$ coefficients (except of the norm of $f$ being sufficiently small). Hence the following equation $e^{i\,\nu_{0\cdots 01}} \,\text{cos}\, \theta_{0\cdots 01} =g(\theta_I,\nu_I : I\in \mathcal{J}_1)$ will be not satisfied except of measure zero subspace for such induced probability distribution. This shows that \cref{conditional5} holds true, and hence finishes the proof of \cref{lemma1odd}.
\end{proof}

\begin{proposition}
\label{propp33}
Consider a generic pure quantum state $\ket{\psi}\in \mathcal{H}^{\otimes n+1}_2$ written in the form $\ket{\psi}= \ket{0}\ket{\psi_0}+\ket{1}\ket{\psi_1}$  where $\ket{\psi_0},\ket{\psi_1}\in \mathcal{H}^{\otimes n}_2$. Furthermore, consider the polynomial equation
\begin{equation}
\label{polypoly3}
E^{(n)} \big(\ket{\psi_0}+z\ket{\psi_1}\big) =0
\end{equation}
in $z$ variable, where $E^{(n)}$ is a SLIP measure presented in \cref{(2k+1)} for an even number of qubits $n=2(k+2)$ and \cref{(2k)} for an odd number of qubits $n=2k+1$. For a generic state $\ket{\psi}\in \mathcal{H}^{\otimes n+1}_2$ the above polynomial equation has exactly four distinct roots with probability one. In other words, the set of states for which \cref{polypoly3} has less then four distinct roots is of measure zero with respect to the Fubini–Study measure. 
\end{proposition}

\begin{proof}[Proof \cref{propp33}]
We shall see that the statement of \cref{propp33} esily follows from \cref{lemma1odd} and \cref{lemma1even} and properties of Fubini–Study measure. Recall that Fubini–Study probability distribution on $\mathcal{H}_2^{\otimes n+1}$ is invariant under any unitary operation acting on the whole space $\mathcal{H}_2^{\otimes n+1}$. Hence the probability that \cref{polypoly3} has multiple roots is bigger or equal to the probability that \cref{polypoly3} has multiple root at $z=0$ under the condition that it has at least single root at $z=0$ is zero. Hence \cref{lemma1odd} and \cref{lemma1even} justify the statement.
\end{proof}

Observe that while using a SLIP measure defined in \cref{(2k+1),(2k)} (depending on the parity of the number of qubits $n$) in Procedure 1 (described in the main body of the paper) a generic pure $n$-qubit state has exactly four roots with probability one. Therefore, together with Proposition 1 from the main body of the paper (proven in \cref{22november} of the Supplementary Material), it shows that a single SLIP measure is enough to provide necessary and sufficient conditions for any two generic pure $n$-qubit states to be SLOCC-equivalent. Hence, we have the following corollary.

\begin{corollary}
\label{easyBusy}
A single SLIP measure, \cref{(2k+1)} for an even number of qubits $n=2(k+2)$ and \cref{(2k)} for an odd number of qubits $n=2k+1$, is enough to provide necessary and sufficient conditions for any two generic pure $n$-qubit states to be SLOCC-equivalent.
\end{corollary}

\noindent
\cref{easyBusy} can be written shortly as Proposition 2 from the main body of the paper.

\section{Unique Möbius transformation}
\label{11november}

It is well known that for a given tuple of three distinct points $z_1,z_2,z_3$ and any second tuple of such points $w_1,w_2,w_3$ on the extended complex plane $\hat{\mathbb{C}}$, there is a unique Möbius transformation $T$, which transforms one tuple into the other, i.e.
\[
T(z_1)=w_1, \quad
T(z_2)=w_2, \quad
T(z_3)=w_3\;
.
\]
There are several ways to determine the form of $T$. An explicit form can be found by evaluating the determinant \cite{TothMTbook}:
\begin{equation}
T(z):=
{\displaystyle
\text{det}
\begin{pmatrix}zw&z&w&1\\z_{1}w_{1}&z_{1}&w_{1}&1\\z_{2}w_{2}&z_{2}&w_{2}&1\\z_{3}w_{3}&z_{3}&w_{3}&1\end{pmatrix}
}
.
\end{equation}
This results into the following form of $T$
\begin{equation}
T(z)=\dfrac{az+b}{cz+d}
\end{equation}
where
\begin{align*}
{\displaystyle a=\det {\begin{pmatrix}z_{1}w_{1}&w_{1}&1\\z_{2}w_{2}&w_{2}&1\\z_{3}w_{3}&w_{3}&1\end{pmatrix}},}
&
\quad
{\displaystyle b=\det {\begin{pmatrix}z_{1}w_{1}&z_{1}&w_{1}\\z_{2}w_{2}&z_{2}&w_{2}\\z_{3}w_{3}&z_{3}&w_{3}\end{pmatrix}},}
\\
\nonumber
{\displaystyle c=\det {\begin{pmatrix}z_{1}&w_{1}&1\\z_{2}&w_{2}&1\\z_{3}&w_{3}&1\end{pmatrix}},}
&
\quad
{\displaystyle d=\det {\begin{pmatrix}z_{1}w_{1}&z_{1}&1\\z_{2}w_{2}&z_{2}&1\\z_{3}w_{3}&z_{3}&1\end{pmatrix}}}
.
\end{align*}
Note that the corresponding $\SLC$ operator is of the form $\mathcal{O}= \dfrac{1}{N} {\begin{pmatrix}a&b\\c&d\end{pmatrix}}$ with the normalization constant
\begin{align*}
N= \sqrt{\det {\begin{pmatrix}a&b\\c&d\end{pmatrix}}}& =
\sqrt{(z_{1}-z_{2})(z_{1}-z_{3})(z_{2}-z_{3})}\cdot\\
&\cdot\sqrt{(w_{1}-w_{2})(w_{1}-w_{3})(w_{2}-w_{3})}
.
\end{align*}

\section{Rotation group $\mathcal{G}_{24}$}
\label{NorFor}
Consider the set of four symmetrically related points $\Phi=\{z,\frac{1}{z},-z,-\frac{1}{z} \}$.
It is very convinient to associate with them the cuboid spanned by eight points:
\[
\Phi\cup\bar{\Phi}=
\Big\{z,\frac{1}{z},-z,-\frac{1}{z} , \bar{z},\frac{1}{\bar{z}},-\bar{z},-\frac{1}{\bar{z}} \Big\},
\]
as it is presented on \cref{G24}. Observe, that all six faces of the cuboid are parallel to one of the planes: $XZ$,$XY$, or $YZ$.
In fact, this property is equivalent to the initial assumption that the set of points $\Phi$ is in normal form.
Clearly, all rotations of the Bloch ball preserve the form of the cuboid.
Nevertheless, only a special subgroup of all rotations preserve faces of the cuboid being parallel to $XZ$,$XY$, or $YZ$.
This special subgroup $\mathcal{G}_{24}$ contains $24$ elements spanned by three rotations of $\pi /2$ around $X$, $Y$, and $Z$ axis, given by:
\begin{align}
\mathbf{R}_x ({\pi}/{2}) = &
\begin{psmallmatrix}\text{cos} \;\pi /4  & -i\;\text{sin}\;\pi /4\\-i\;\text{sin}\;\pi /4 & \text{cos}\;\pi /4\end{psmallmatrix}
=\frac{1}{\sqrt{2}}
\begin{psmallmatrix}1 & -i\\-i &1 \end{psmallmatrix}
,\; \label{Indeed} \\
\mathbf{R}_y ({\pi}/{2}) = &
\begin{psmallmatrix}\text{cos} \;\pi /4  & -\text{sin}\;\pi /4\\\text{sin}\;\pi /4 & \text{cos}\;\pi /4\end{psmallmatrix}
=\frac{1}{\sqrt{2}}
\begin{psmallmatrix}1 & -1\\1 &1 \end{psmallmatrix}
,\; \label{Indeed2} \\
\mathbf{R}_z ({\pi}/{2}) = &
\begin{psmallmatrix}e^{-i \pi /4 } & 0\\0 & e^{i \pi /4 }\end{psmallmatrix}
=\frac{1}{\sqrt{2}}
\begin{psmallmatrix}1-i  & 0\\0 & 1+i \end{psmallmatrix} \label{Indeed3}
\end{align}
In fact, this is a group of rotations preserving the regular cube (the group of orientable cube symmetries).
Clearly, all rotations in the $\mathcal{G}_{24}$ group preserve the normal-form structure of $\Phi$.
On the other hand, the normal form is uniquelly determined up to $24$ rotations in the $\mathcal{G}_{24}$ group.

\section{Proof of \cref{24}}
\label{secNorForm}

We present a proof of \cref{24}. 
For each complex number $\lambda$ there exists another complex number $z$, such that the cross-ratio of the four points is equal to $\lambda$, i.e.
\begin{equation}
\Big(z,\frac{1}{z};-z,-\frac{1}{z} \Big)=
\lambda\,.
\label{CRvia}
\end{equation}
Indeed, the cross-ratio on the left side equals ${4 z^2} /{(1+z^2)^2}$, and the equation ${4 z^2} /{(1+z^2)^2}=\lambda$ has exactly four solutions
\begin{equation}
\label{eq23}
z_0 =\frac{4-2\lambda + \sqrt{1-\lambda}}{2\lambda}
,\; \frac{1}{z_0},\;-z_0,\;-\frac{1}{z_0}.
\end{equation}
Therefore, for a given value $\lambda$ there exists a unique $\lambda$-normal system, such that the cross-ratio of its vertices is given by $(z_0,\frac{1}{z_0};-z_0,-\frac{1}{z_0} )=\lambda$.
Replacing the vertex $z_0$ by any other vertex $z_0,{1}/{z_0},-z_0$, or $-{1}/{z_0}$ does not change the value of the cross-ratio $(z_0,\frac{1}{z_0};-z_0,-\frac{1}{z_0} )=\lambda$. Note that there exists a unique Möbius transformation $T$ which maps $z_1,z_2,z_3$ onto $z_0, {1}/{z_0},-z_0$, with the remaining $z_4$ mapped onto $-{1}/{z_0}$. Observe as well that the value of $z_0$ is unique up to its inverse, opposite and opposite inverse elements, according to \cref{eq23}, with the corresponding Möbius transformations associated to the matrices $T, \sigma_x T ,\sigma_y T$ and $\sigma_z T$.
Each of those transformations maps the set of points $\{z_1,z_2,z_3,z_4\}$ onto the same set of points $\{z_0, {1}/{z_0},-z_0,-{1}/{z_0}\}$, although the exact bijection between those two sets of roots is different for each transformation.

Depending on the order of four points $\{z_1,z_2,z_3,z_4\}$, the corresponding cross-ratio takes six values:
$\lambda, \frac{1}{\lambda}, 1-\lambda ,\frac{1}{1-\lambda},\frac{\lambda-1}{\lambda}$ and $\frac{\lambda}{\lambda-1}$.
For each of these, there is a corresponding set of solutions of the form $\{z_0, {1}/{z_0},-z_0,-{1}/{z_0}\}$ via \cref{eq23} with four related Möbius transformations. Therefore, there are in total 24 Möbius transformations that map any four non-degenerated points onto a normal system, each of them related by an element of the group $\mathcal{G}_{24}$ which has exactly $24$ elements.


\section{Transformation of a state into its Normal form}
\label{NorForEx}
We illustrate the procedure to determine the normal form of 4-qubit states  on the following example:
\begin{align}
\label{1101}
\ket{\psi}&\propto
\ket{0}\ket{\text{GHZ}}+\ket{1}\ket{\text{W}}\\
\nonumber
&=
\ket{0000}+\ket{0111}+\ket{1100}+\ket{1010}+\ket{1001} \\
\nonumber
&=
\ket{0}\underbrace{\ket{000}+\ket{111}}_{\ket{\psi_0}}+
\ket{1}\underbrace{\ket{100}+\ket{010}+\ket{001}}_{\ket{\psi_1}}.
\end{align}
of the widely discussed four-partite state \cite{Osterloh3,OsterlohWernerStates}. We shall focus attention on the first subsystem. Corresponding states $\ket{\psi_0}$ and $\ket{\psi_1}$ are indicated on \cref{1101}. We shall use the 3-tangle measure. The following polynomial
$\tau{(3)} (\ket{\psi_0}+z \ket{\psi_1})$
is of a degree four in variable $z$, and has exactly four distinct roots:
\begin{align*}
z_1 &=0, \\
z_2 &=-\sqrt[3]{4}, \\
z_3 &=-\sqrt[3]{4} \; e^{2 \pi i /3}, \\
z_4 &=-\sqrt[3]{4} \; e^{4 \pi i /3} .
\end{align*}
The corresponding value of the cross-ratio is equal to $\lambda = 1+ e^{4 \pi i /3}$. As it is shown in \cref{24}, there is a unique Möbius transformation $T$ which maps $z_1,z_2,z_3,z_4$ onto the system $z_0,\frac{1}{z_0},-z_0,-\frac{1}{z_0} $ where $z_0$ is a root of a polynomial ${4 z^2} /{(1+z^2)^2}=\lambda$. Choose one of its roots, e.g. $z_0= 1+ \sqrt{2}$. \cref{22Jun} presents one way of obtaining the exact transformation $T$, giving
\begin{equation}
T(z)=
\dfrac{z-\frac{1+\sqrt{3}}{\sqrt[3]{4}}}{\frac{(1+\sqrt{3})(1+i)}{2} z+\frac{1+i}{\sqrt[3]{4}}}
\end{equation}
which performs the mapping
\[
T(z_1)=z_0= 1+ \sqrt{2}, \quad
T(z_2)=\frac{1}{z_0}, \quad
T(z_3)=-z_0.
\]
According to \cref{24}, $T(z_4)= -\frac{1}{z_0}$ and hence $T$ maps $z_1,z_2,z_3,z_4$ into the normal system of roots. Note that the corresponding $\SLC$ operator is of the form
\begin{equation}
\label{eqww}
{\displaystyle
\mathcal{O}_1 =
\dfrac{\sqrt[3]{2}}{(1+i) (3+\sqrt{3})}
{\begin{pmatrix}
1&
-\frac{1+\sqrt{3}}{\sqrt[3]{4}}\\
\frac{(1+\sqrt{3})(1+i)}{2}&
\frac{1+i}{\sqrt[3]{4}}
\end{pmatrix}}
,
}
\end{equation}
which is presented on \cref{MobiusTrans}. 
Similar calculations for three remaining subsystems lead to the following SLOCC operator
\begin{equation}
\mathcal{O} =\mathcal{O}_1 \otimes \mathcal{O}_2 \otimes\mathcal{O}_3 \otimes \mathcal{O}_4
\end{equation}
where $\mathcal{O}_1$ is presented on \cref{eqww} and $\mathcal{O}_2 = \mathcal{O}_3 =\mathcal{O}_4 := \mathcal{O}_1 \sigma_x$, which transforms state (\ref{1101}) into the $\ket{\mathcal{G}_{abcd}}$ form
\begin{align*}
\ket{G_{abcd}} &=\\
&
\tfrac{a+d}{2} \big(\ket{0000}+\ket{1111} \big)+
\tfrac{a-d}{2} \big(\ket{0011}+\ket{1100} \big)
\\
+&
\tfrac{b+c}{2}  \big(\ket{0101}+\ket{1010} \big)+
\tfrac{b-c}{2} \big(\ket{0110}+\ket{1001} \big)
\end{align*}
with parameters
\begin{align*}
a=& (-8 + 4 i) - \frac{(12 - 8 i)}{\sqrt{3}},\\
b=& \frac{8}{3} i (3 + 2 \sqrt{3}),\\
c=& 0,\\
d=& -\frac{4}{3} ((6+3i) + (3+2i) \sqrt{3}).
\end{align*}

\section{Proof of \cref{propGabcd}}
\label{GabcdProof}

Consider the state $\ket{G_{abcd}}$ and its decomposition with respect to the first subsystem $\ket{G_{abcd}} =\ket{0}\ket{\psi_0} +\ket{1}\ket{\psi_1}$, where
\begin{align*}
\ket{\psi_0} &= \tfrac{a+d}{2}\ket{000} +\tfrac{a-d}{2}\ket{011} +\tfrac{b+c}{2}\ket{101} +\tfrac{b-c}{2}\ket{110}\,, \\
\ket{\psi_1} &= \tfrac{a+d}{2}\ket{111} +\tfrac{a-d}{2}\ket{100} +\tfrac{b+c}{2}\ket{010} +\tfrac{b-c}{2}\ket{001}\,.
\end{align*}
Suppose that $\tau^{(3)} (\zeta \ket{\psi_0} +\ket{\psi_1} )=0$.
Since $\tau^{(3)}$ is a $\SLC^{\otimes 3}$ invariant, for any local operators $\mathcal{O}_1$, $\mathcal{O}_2$, $\mathcal{O}_3$ we have
\[
\tau^{(3)} \Big( (\mathcal{O}_1\otimes\mathcal{O}_2\otimes\mathcal{O}_3)
\big(\zeta\ket{\psi_0} +\ket{\psi_1} \big)\Big)=0\,.
\]
Observe that
\begin{align*}
\ket{\psi_0}=&(\sigma_x \otimes \sigma_x \otimes \sigma_x) \ket{\psi_1}\,,\\
\ket{\psi_1}=&(\sigma_x \otimes \sigma_x \otimes \sigma_x) \ket{\psi_1}\,,
\end{align*}
where $\sigma_x,\sigma_y,\sigma_z$ are Pauli matrices.
Therefore by taking all local operators $\mathcal{O}_1, \mathcal{O}_2, \mathcal{O}_3$ equal to $\sigma_x$, we may conclude that
\begin{align*}
0=\tau^{(3)} \Big( (\sigma_x &\otimes \sigma_x \otimes \sigma_x)
\big(\zeta\ket{\psi_0} +\ket{\psi_1} \big)\Big)=
\\
\nonumber
&=
\zeta\ket{\psi_1} +\ket{\psi_0}
\propto
\frac{1}{\zeta}\ket{\psi_0}+\ket{\psi_1},
\end{align*}
hence $1/\zeta$ is another root of $\tau^{(3)}$. Similarly, by considering $(\sigma_y \otimes \sigma_y \otimes \sigma_y)$ and $(\sigma_z \otimes \sigma_z \otimes \sigma_z)$, one may find another two roots $-\zeta, \;-1/\zeta$ of $\tau^{(3)}$.
This shows that the roots of $\tau^{(3)}$ evaluated on any state from the $G_{abcd}$ family are symmetrical with respect to rotations around $X,Y,Z$ axes by the angle $\pi$.
Writting $\tau^{(3)} (z \ket{\psi_0} +\ket{\psi_1} )=0$ explicitely, we obtain the equation
\begin{equation*}
\tau^{(3)} (z \ket{\psi_0} +\ket{\psi_1} ) = A z^4 - 2(2 B+A)z^2 + A = 0,
\label{eqq}
\end{equation*}
where $A = (b^2 - c^2)  (a^2 - d^2)$ and $B = (c^2-d^2)  (a^2 - b^2)$. The above equation is non-degenerated iff $A,B,A+2B\neq 0$, which happens iff $a^2 \neq b^2 \neq c^2 \neq d^2 $ are pairwise different.

\begin{lemma}
Any local operator
$\mathcal{O}=
\mathcal{O}_1 \otimes \mathcal{O}_2 \otimes \mathcal{O}_3 \otimes \mathcal{O}_4 \in \SLC^{\otimes 4}$
which transforms states
$\ket{G_{a' b'c'd'}} \propto \mathcal{O} \ket{G_{a bcd}}$ with $a^2 \neq b^2 \neq c^2 \neq d^2 $, is of the form $\mathcal{O}_i \in \mathcal{G}_{24} $.
\end{lemma}
\begin{proof}
A local operator $\mathcal{O}_1$ acting on the first qubit and transforming the state $ \ket{G_{abcd}}$ onto $\ket{G_{a'b'c'd'}}$, also transforms their systems of roots denoted as $\Lambda$ and $\Lambda '$, respectively, via the action of the corresponding Möbius transformation. Note that both systems $\Lambda$ and $\Lambda '$ are in the normal form, therefore, according to \cref{24}, we have that $\mathcal{O}_i \in \mathcal{G}_{24} $. A similar analysis with respect to all other qubits shows that $\mathcal{O}_2,\mathcal{O}_3,\mathcal{O}_4 \in \mathcal{G}_{24}$.
\end{proof}

This way, searching for SLOCC-equivalence between the states $\ket{G_{a bcd}}$ and $\ket{G_{a' b'c'd'}}$ becomes restricted to the search within the finite class of operators $\mathcal{O} \in \mathcal{G}_{24}^{\otimes 4}$. Since the group $\mathcal{G}_{24}$ has only 24 elements, one may numerically verify that there are exactly $8\times 24 =192$ states in the $G_{abcd}$ family which are SLOCC-equivalent to $\ket{G_{abcd}}$ by $\mathcal{O}\in \mathcal{G}_{24}^{\otimes 4}$.
For example, the following operation
\begin{align}
\label{tuples1}
\mathbf{R}_x (\tfrac{\pi}{2}) \otimes
\mathbf{R}_x (\tfrac{\pi}{2})\otimes
\mathbf{R}_x (\tfrac{\pi}{2}) \otimes
\mathbf{R}_x (\tfrac{\pi}{2}) &
\end{align}
transforms state $\ket{G_{abcd}}$ into $\ket{G_{-b-acd}}$.
This might be simply written as a transformation of a tuples of indices: the tuple $(a,b,c,d )$ is transformed into the tuple $(-b,-a,c,b)$. Similarly, the operators showed on the following right hand sides provide the corresponding transformations of the tuple $(a,b,c,d)$ on the left side:
\begin{align}
\label{tuples2}
\nonumber
\mathbf{R}_y (\tfrac{\pi}{2}) \otimes
\mathbf{R}_y (\tfrac{\pi}{2})\otimes
\mathbf{R}_y (\tfrac{\pi}{2}) \otimes
\mathbf{R}_y (\tfrac{\pi}{2}) &\, \longleftrightarrow \, (a,\blue{d},c,\blue{b}), \\
\nonumber
\mathbf{R}_z (\tfrac{\pi}{2}) \otimes
\mathbf{R}_z (\tfrac{\pi}{2})\otimes
\mathbf{R}_z (\tfrac{\pi}{2}) \otimes
\mathbf{R}_z (\frac{\pi}{2}) &\, \longleftrightarrow \, (\blue{-d},b,c,\blue{-a}), \\
\nonumber
\mathbf{R}_y ({\pi}) \otimes
\mathbf{R}_y ({\pi})\otimes
\bs{1} \otimes
\bs{1} &\, \longleftrightarrow \,
(a,\blue{-b},\blue{-c},d), \\
\nonumber
\mathbf{R}_x ({\pi}) \otimes
\mathbf{R}_x ({\pi})\otimes
\bs{1} \otimes
\bs{1} &\, \longleftrightarrow \,
(a,b,\blue{-c},\blue{-d}), \\
\nonumber
\mathbf{R}_y ({\pi}) \otimes
\bs{1} \otimes
\mathbf{R}_y ({\pi})\otimes
\bs{1} &\, \longleftrightarrow \,
(\blue{d},\blue{c},\blue{b},\blue{a}), \\
\mathbf{R}_x ({\pi}) \otimes
\bs{1} \otimes
\mathbf{R}_x ({\pi})\otimes
\bs{1} &\, \longleftrightarrow \,
(\blue{b},\blue{a},\blue{d},\blue{c})\,. \nonumber
\end{align}
Additionally, the tuples $ (a,b,c,d)$ and $(-a,-b,-c,-d)$ represent the same state. Note that any composition of the above operations also provides SLOCC equivalences between $\ket{G_{abcd}}$ states. The eight aforementioned transformations of tuples generate all permutations of the $a,b,c,d$ indices, together with the change of a sign of any two or all four indices. There are exactly $24$ permutations and for each permutation the signs can be matched in exactly $1+ {{4}\choose{2}} +1 =8$ ways. This gives in total $192$ tuples representing SLOCC equivalent states, which perfectly matches the numerical result.

Finally, another trivial manipulation with indices $a,b,c,d$ comes from multiplying by a global phase, which is an irrelevant operation due the fact that quantum states are elements of a projective space. This operation transforms the indices as
\[
(e^{i\theta} a,\;e^{i\theta} b,\; e^{i\theta} c,\; e^{i\theta} d ) \sim
(a,b,c,d)\,,
\]
resulting in the same quantum state for any real number $\theta \in [0,2 \pi)$. In particular, for $\theta=\pi$, we observe that system of opposite indices determines the same state as the initial one, i.e. $(-a,-b,-c,-d)\sim (a,b,c,d)$.

\section{Exact scenarios}
\label{24September}

\begin{table}[h!]
\renewcommand{\arraystretch}{1.7}
\begin{tabular}{|c||c|}
\hline
Number of
&
Value of  \\
subsystem
& cross-ratio \\ \hline
1& 0.8656 + 0.5008 i \\ \hline
2& 0.5 + 0.5906 i \\ \hline
3& 0.5 - 0.4844 i \\ \hline
4& 0.5 + 0.5397 i \\ \hline
5& 0.9231 + 0.3845 i \\ \hline
\end{tabular}
\caption{Cross-ratio of roots of $M$ measure calculated for each subsystem of $\ket{\psi} $ state rounded to the fourth decimal place.}
\label{TableCR}
\end{table}

In order to exemplify the viability of the results in this work, we present another non-trivial scenarios where our method is useful. 
We show how a single 4-qubit entanglement $\SLIP_4^4$ measure might be used to verify whether 5-qubit states are SLOCC equivalent. Among several $\SLIP_4^h$ measures defined for system of four qubits, the so-called $M$ measure has degree $4$ \cite{PolynomialInvariantsofFourQubits,Djokovic4qubits} and is defined as the determinant
\begin{equation}
M(\ket{\psi}):=
{\displaystyle
\text{det}
\begin{pmatrix}
c_{0000}&c_{1000}&c_{0010}&c_{1010}\\
c_{0001}&c_{1001}&c_{0011}&c_{1011}\\
c_{0100}&c_{1100}&c_{0110}&c_{1110}\\
c_{0101}&c_{1101}&c_{0111}&c_{1111}
\end{pmatrix}
}
\end{equation}
where $c_{i_1,i_2,i_3,i_4}$ are state coefficients, i.e
\[
\ket{\psi}=  \sum_{i_1,i_2,i_3,i_4=0}^1
c_{i_1,i_2,i_3,i_4} \ket{i_1,i_2,i_3,i_4}
.
\]
 We demonstrate now how the $M$ measure can be used to verify whether two given 5-qubit states are SLOCC-equivalent. As an example, consider the following 5-qubit state
\begin{align*}
\ket{\psi} =&
\ket{0 0 0 0 1}+
\ket{0 0 1 1 1} +
\ket{0 1 0 0 1} +
\ket{0 1 0 1 0} +
\ket{0 1 1 0 1} +\\&
\ket{1 0 0 0 0} +
\ket{1 0 0 1 0} +
\ket{1 0 1 0 1} +
\ket{1 0 1 1 0} +
\ket{1 1 0 1 1} +\\&
\frac{1}{2} \ket{1 1 1 0 0}.
\end{align*}
By permuting the five subsystems of the state $\ket{\psi}$, we obtain in total $5!=120$ five qubit states. In the following, we will show that such states are not pairwise SLOCC-equivalent. Indeed, for each subsystem of a state $\ket{\psi} $, we may calculate the roots of the $M$ measure according to Procedure 1. For any subsystem, there are four distinct roots. The values of the corresponding cross-ratios are presented in \cref{TableCR}. It is straightforward to show that no two values of cross-ratios corresponding to different subsystems (see \cref{TableCR}) are related by $\lambda, \frac{1}{\lambda}, 1-\lambda ,\frac{1}{1-\lambda},\frac{\lambda-1}{\lambda}$ or $\frac{\lambda}{\lambda-1}$. As a direct consequence of Theorem 1, one concludes that those states are not SLOCC-equivalent. Indeed, the states obtained from $\ket{\psi} $ by non-identical permutation of subsystems will have corresponding cross-ratios simultaneously permuted. Therefore, for some subsystems, the values of cross-ratios corresponding to any two such states are not related via $\lambda, \frac{1}{\lambda}, 1-\lambda ,\frac{1}{1-\lambda},\frac{\lambda-1}{\lambda}$ or $\frac{\lambda}{\lambda-1}$. As a consequence, there is no Möbius transformation relating the roots of any subsystem (see \cref{1Dec}), and hence there is no SLOCC operator which might relate those two states. A similar reasoning holds true for any pair of states obtained from $\ket{\psi}$ by permuting its subsystems by two distinct permutations.

Finally, we note that these results would otherwise be highly computationally demanding if standard procedures were applied.

\bibliography{Physics}

\begin{thebibliography}{47}
\bibitem [{com()}]{commentSLIPgenerate}%
  \BibitemOpen
  \href@noop {} {}\bibinfo {note} {The coefficients of states were
  independently randomly chosen complex numbers with normal distribution N(0,1)
  for imaginary and real part. A sampling of (n+1)-qubit states and calculating
  their entanglement measure was swift on a PC computer for
  n$<$20.}\BibitemShut {Stop}%
\end{thebibliography}
\end{document}